\documentclass[a4paper,11pt]{amsart} 
\usepackage{amsmath}
\usepackage{amsfonts}
\usepackage{amssymb}
\usepackage{amsthm}
\usepackage{newlfont}
\usepackage{color}
\usepackage[pdftex]{hyperref}
\usepackage{subfig}
\usepackage{bbm}
\usepackage{cite} 
\usepackage{transparent} 
\usepackage{mathrsfs}
\usepackage{fancyhdr}
\usepackage{graphicx}
\usepackage{geometry}
\usepackage{psfrag}
\usepackage{textcomp}

\let\eps=\varepsilon

\newcommand{\caB}{{\mathcal B}}
\newcommand{\caC}{{\mathcal C}}

\newcommand{\caE}{{\mathcal E}}
\newcommand{\caF}{{\mathcal F}}

\newcommand{\caO}{{\mathcal O}}
\newcommand{\caP}{{\mathcal P}}

\newcommand{\caS}{{\mathcal S}}

\newcommand{\bbE}{{\mathbb E}}

\newcommand{\bbN}{{\mathbb N}}

\newcommand{\bbP}{{\mathbb P}}

\newcommand{\bbR}{{\mathbb R}}

\newcommand{\bbT}{{\mathbb T}}

\newcommand{\bbZ}{{\mathbb Z}}

\newcommand{\dd}{\mathrm{d}}

\newcommand{\beq}{\begin{equation}}
\newcommand{\eeq}{\end{equation}}
\newcommand{\beqn}{\begin{equation*}}
\newcommand{\eeqn}{\end{equation*}}
\newcommand{\baq}{ \begin{eqnarray} }
\newcommand{\eaq}{ \end{eqnarray} }

\newcommand{\und}[1]{\underline{#1}}
\newcommand{\ampl}{\mathrm{Amp}}

\newcommand{\oln}{\overline{n}}
\newcommand{\tthree}{\bbT^3}
\newcommand{\ttwo}{\bbT^2}
\newcommand{\upm}{^{\mathrm{main}}}
\newcommand{\upo}{^{(1)}}
\newcommand{\upt}{^{(2)}}
\newcommand{\vari}[1]{\mathrm{Var}\left(#1\right)}
\newcommand{\bra}{\left\langle}
\newcommand{\ket}{\right\rangle}
\newcommand{\covari}[1]{\mathrm{Cov}\left(#1\right)}
\newcommand{\bigpi}{\Pi^{conn}_{n_1,n_2}}

\newtheorem{thm}{Theorem}
\newtheorem{cor}[thm]{Corollary}

\newtheorem{lma}[thm]{Lemma}
\newtheorem{defin}{Definition}

\newcommand{\dist}{\mathrm{dist}}

\begin{document}

\title[Self-averaging]{Dynamical self-averaging for a lattice Schr\"{o}dinger equation with weak random potential}

\author[M. Butz]{Maximilian Butz}
\address{Fakult\"at f\"ur Mathematik \\ Technische Universit\"at M\"unchen \\
Boltzmannstr.~3 \\ 
85748 Garching, Germany}
\email{maximilianjbutz@gmail.com}

\date{\today }

\begin{abstract}
We study the kinetic, weak coupling limit of the dynamics governed by a discrete random Schr\"{o}dinger operator on $\bbZ^3$. For sequences of $\ell^2\left(\bbZ^3\right)$-bounded initial states and convergent initial Wigner transform, we prove that the scaled Wigner transform converges to the solution of a linear Boltzmann equation in $L^r\left(\bbP\right)$for all $r>0$, thus considerably strengthening a previous result by Chen. The key ingredients for the proof are a finer classification of graphs in the expansion of the perturbed dynamics as well as a novel resolvent estimate for the unperturbed Schr\"{o}dinger operator. Under some additional assumption on the sequence of initial states we even prove almost sure convergence.
\end{abstract}

\maketitle

\section{Introduction}
\noindent{}To understand the motion of independent electrons in a background lattice with random impurities, 
Anderson \cite{anderson} proposed as simplified model a single particle Sch\"{o}dinger equation 
with random Hamiltonian
\begin{equation}
\label{hamiltonian}
H_\omega=-\frac{1}{2}\Delta+\lambda V_\omega.
\end{equation}
$H_\omega$ acts on $\ell^2\left(\bbZ^d\right)$, $\Delta$ denotes the nearest-neighbor lattice Laplacian
\begin{equation}
(\Delta\psi)(x)=-2d\psi(x)+\sum_{y:|y-x|=1}\psi\left(y\right),
\end{equation}
and $V$ is a random potential given by 
\begin{equation}
V_\omega(x)=\omega_x
\end{equation}
with $\left(\omega_x\right)_{x\in\bbZ^d}$ independent centered Gaussian random variables with $\bbE\left[\omega_x^2\right]=1$. The strength of the potential is controlled by a coupling constant $\lambda>0$. The goal is to understand the properties of the time evolution governed by
\begin{equation}
\label{schrodinger_equation_1}
i\frac{\dd}{\dd t}\phi_t=H_\omega\phi_t
\end{equation}
with prescribed initial conditions $\phi_0\in\ell^2\left(\bbZ^d\right)$, $\Vert\phi_0\Vert_{\ell^2}=1$. 

In our contribution we study the dynamics (\ref{schrodinger_equation_1}), in the kinetic limit, i.e. with a weak coupling constant $\lambda\ll1$ and on kinetic space and time scales. A very useful tool in this analysis is the Wigner transform $W[\phi_t](x,v)$ of the wave function $\phi_t$, (to be defined in Section \ref{main_result}) which takes values in phase space ($x\in\bbZ^d/2$, $v\in[0,1)^d=\bbT^d$). As the wave function scatters off the random potential $V_\omega$, $\phi_t$ and thus $W$ will be random as well. We expect a finite number of scattering events to occur on space and time scales of order $\lambda^{-2}$, thus define $X=\lambda^2 x$, $V=v$, $T=\lambda^2t$, and test the Wigner function against a function $J$ living on the kinetic phase space $\bbR^d\times\bbT^d$,
\begin{equation}
\label{intro_wtest}
\sum_{X\in(\lambda^2\bbZ/2)^d}\int_{\bbT^d}\dd V \overline{J(X,V)}W\left[\phi_{\lambda^{-2}T}\right]\left(\lambda^{-2}X,V\right),
\end{equation}
which is still random. For $d=3$, the expectation of this quantity, i.e. the \emph{disorder-averaged} Wigner transform, can be shown to converge to
\begin{equation}
\label{intro_mutest}
\int_{\bbR^3\times\bbT^3}\overline{J(X,V)}\mu_T(\dd X, \dd V)
\end{equation}
in the $\lambda\rightarrow0$ limit, $\mu_T$ the solution of a linear Boltzmann equation \cite{chen_rs}, by the graph expansion method that was originally developed to handle the continuous analogon \cite{erdyau}. Naturally, it would be very desirable to not only understand how the expectation of the Wigner transform behaves in the kinetic limit, but to also have a control on the typical deviations from this average. Ideally, one would hope for a law of large numbers-type, almost sure convergence statement. In this paper, we prove that in the kinetic scaling for the dynamics governed by (\ref{schrodinger_equation_1}), the Wigner transform in fact converges in higher mean, i.e.
\begin{equation}
\label{intro_highermeanconv}
\begin{split}
\lim_{\lambda\rightarrow0}&\bbE\left[\left|\sum_{X\in(\lambda^2\bbZ/2)^3}\int_{\bbT^3}\dd V \overline{J(X,V)}W\left[\phi_{\lambda^{-2}T}\right]\left(\lambda^{-2}X,V\right)\right.\right.\\
&\hspace{3.5cm}\left.\left.-\vphantom{\sum_{X\in(\lambda^2\bbZ/2)^3}\int_{\bbT^3}\dd V \overline{J(X,V)}W\left[\phi_{\lambda^{-2}T}\right]\left(\lambda^{-2}X,V\right)}\int_{\bbR^3\times\bbT^3}\overline{J(X,V)}\mu_T(\dd X, \dd V)\right|^{\displaystyle r}\right]=0
\end{split}
\end{equation}
for all $r>0$, implying convergence of (\ref{intro_wtest}) to (\ref{intro_mutest}) in probability. For WKB initial states, we even prove that 
\begin{equation}
\label{deviations}
\left|\sum_{X\in(\lambda^2\bbZ/2)^3}\int_{\bbT^3}\dd V \overline{J(X,V)}W\left[\phi_{\lambda^{-2}T}\right]\left(\lambda^{-2}X,V\right)-\int_{\bbR^3\times\bbT^3}\overline{J(X,V)}\mu_T(\dd X, \dd V)\right|
\end{equation}
vanishes almost surely as $\lambda\rightarrow0$ and uniformly for macroscopic times $T$ from compact intervals. In this sense, the time-evolved scaled Wigner transform converges to a deterministic limit under kinetic scaling, and the linear Boltzmann equation is a valid approximation for the rescaled Schr\"{o}dinger dynamics for almost every realization of the potential $V_\omega$. This phenomenon is usually referred to as \emph{self-averaging}. 

The convergence (\ref{intro_highermeanconv}) has previously been shown in \cite{chen_lp}, but under concentration of singularity assumptions on the initial data $\phi_0$ which are technical and generally hard to check for a particular initial condition. For related results on the classical Lorentz gas \cite{spohn_revmodphys,bbs}, self-averaging requires some randomness in the initial condition. Quantum mechanically, any wave function carries some intrinsic randomness. Differing from what was expected previously, our result shows that actually any given $\ell^2$-bounded initial state carries enough randomness to self-average on kinetic time and space scales, so (\ref{intro_highermeanconv}) holds in full generality.

In our proof, we follow the ideas of \cite{chen_lp} and use the technique of graph expansions of \cite{erdyau}. The occuring graphs are classified by certain structures, and for each class there are estimates that show that their contributions vanish in the kinetic limit, thus complementing the bounds that have so far been found for special cases in \cite{chen_lp}. To obtain those estimates, we have to derive a novel bound for certain resolvent integrals, which is complicated by the more intricate geometry of the level surfaces resulting from the dispersion relation of the lattice model.

\vspace{5mm}
A natural next step would be to derive a central limit theorem, i.e. to rescale the fluctuations in (\ref{deviations}) such that they stay $\caO(1)$ in the $\lambda\rightarrow0$ limit and to study their behavior on kinetic time scales. For the case of a random potential decorrelating in time \cite{clt}, these fluctuations have been shown to scale like $\lambda$, and, after rescaling, to converge to a solution of the linear Boltzmann equation that also governs the main term, however with random initial conditions.

\section{Main result}
\label{main_result}
From now on, we set $d=3$, but our results extend to $d>3$. For a fixed $\eta=\lambda^2$ let $\phi^{(\eta)}_t$, $t\geq0$, denote the solution of the Schr\"{o}dinger equation
\begin{equation}
\label{schrodinger_equation}
i\frac{\dd}{\dd t}\phi=H_\omega\phi
\end{equation}
started with a $\phi^{(\eta)}_0$ from a sequence $\left(\phi^{(\eta)}_0\right)_{\eta>0}$ of initial states in $\ell^2\left(\bbZ^3\right)$. As a physically minimal condition we impose
\begin{equation}
\label{boundedprob}
\sup_{\eta>0}\left\Vert\phi^{(\eta)}_0\right\Vert^2_{\ell^2\left(\bbZ^3\right)}\leq 1.
\end{equation}
For positive times, the quantum particle has interacted with the random potential, so $\phi^{(\eta)}_t$ is a random wave function depending on $\omega$, a dependence which will be notationally suppressed throughout this paper. 

For the Fourier transform $\ell^2\left(\bbZ^3\right)\rightarrow L^2\left(\tthree\right)$, with $\bbT^3=[0,1)^3$ denoting the three-dimensional unit torus, we will use the unitary extension of
\begin{equation}
\widehat{\psi}(k)=\sum_{x\in\bbZ^3}\psi(x)e^{-2\pi i k\cdot x}.
\end{equation}
One can rewrite the unperturbed Schr\"{o}dinger operator $H_0=-\frac{1}{2}\Delta$ in momentum space as
\begin{equation}
\widehat{H_0\psi}(k)=e(k)\widehat{\psi}(k)
\end{equation}
with the kinetic energy now a multiplication operator,
\begin{equation}
e(k)=3-\sum_{j=1}^3\cos{2\pi k_j}
\end{equation}
for $k=\left(k_1,k_2,k_3\right)\in\tthree$.

Accordingly, the Fourier transform $L^2\left(\bbR^3\right)\rightarrow L^2\left(\bbR^3\right)$ is the unitary extension of
\begin{equation}
\widehat{f}(k)=\int_{\bbR^3}\dd x f(x)e^{-2\pi i k\cdot x}.
\end{equation}
For Schwartz functions $J\in\caS\left(\bbR^3\times\bbT^3\right)$, $\widehat{J}$ will denote the Fourier transform in the first variable, and we use the shorthand $\widehat{J_\eta}(\xi,v)=\eta^{-3}\widehat{J}(\xi/\eta,v)$.

To study the kinetic limit of (\ref{schrodinger_equation}), we introduce the kinetic space, velocity and time scaling 
\begin{equation}
\label{scaling}
(X,V,T)=(\eta x, v, \eta t)
\end{equation}
and the $\eta$-scaled Wigner transform of $\phi^{(\eta)}_t$
\begin{equation}
W^{\eta}\left[\phi^{(\eta)}_t\right](X,V)=\sum_{\substack{y,z\in\bbZ^3\\y+z=2X/\eta}}\overline{\phi^{(\eta)}_t(y)}\phi^{(\eta)}_t(z)e^{2\pi i V\cdot(y-z)}
\end{equation}
which takes arguments $X\in(\eta\bbZ/2)^3$ and $V\in\tthree$.

The Wigner transform acts as a distribution on Schwartz functions $J\in\caS\left(\bbR^3\times\bbT^3\right)$ by
\begin{equation}
\label{defjw}
\begin{split}
\bra J,W^{\eta}\left[\phi^{(\eta)}_t\right]\ket&=\sum_{X\in(\eta\bbZ/2)^3}\int_{\tthree}\dd V \overline{J(X,V)}W^{\eta}\left[\phi^{(\eta)}_t\right](X,V)\\
&=\int_{\bbR^3\times\tthree}\dd\xi\dd v\overline{\widehat{J_\eta}(\xi,v)}\overline{\widehat{\phi^{(\eta)}_t}(v-\xi/2)}\widehat{\phi^{(\eta)}_t}(v+\xi/2).
\end{split}
\end{equation}
Note that $\widehat{\phi^{(\eta)}_t}$ is defined on $\tthree$ and the last line of (\ref{defjw}) only makes sense if we extend $\widehat{\phi^{(\eta)}_t}$ periodically to $\bbR^3$. While the Wigner transform is quadratic in its argument $\widehat{\phi^{(\eta)}_t}$, it is sometimes useful to have a bilinear Wigner transform of two different states, which is defined in Appendix \ref{app_basic}, Definition \ref{bilinear}. An introduction to lattice Wigner transforms consistent with our definition is given in \cite{luspo}. Whenever (\ref{boundedprob}) holds, we can apply a compactness argument (Alaoglu-Bourbaki theorem together with Lemma \ref{est_bilin}) and the Bochner-Schwartz theorem (as explained in \cite{luspo}) to see that along a subsequence of $\eta\rightarrow0$ (which will not be relabled), the scaled time zero Wigner functions weakly converge to a bounded positive Borel measure $\mu_0$ on $\bbR^3\times\bbT^3$,
\begin{equation}
\label{initialwignerconv}
\lim_{\eta\rightarrow0}\bra J,W^{\eta}\left[\phi^{(\eta)}_0\right]\ket=\int_{\bbR^3\times\bbT^3}\overline{J(X,V)}\mu_0(\dd X, \dd V)=\bra J,\mu_0\ket
\end{equation}
for all $J\in\caS\left(\bbR^3\times\bbT^3\right)$.

Surprisingly, already under the minimal assumption (\ref{boundedprob}), there is a convergence result for the expectation (i.e. the average over all realizations of the potential) of the scaled Wigner function for any positive macroscopic time $T=\eta t$.
\begin{thm}
\label{theorem_convergence_exp}
Let $\left(\phi^{(\eta)}_0\right)_{\eta>0}$ be a sequence of initial states such that (\ref{boundedprob}) and (\ref{initialwignerconv}) hold. Then for all macroscopic times $T>0$, and all test functions $J\in\caS\left(\bbR^3\times\bbT^3\right)$,
\begin{equation}
\lim_{\eta\rightarrow0}\bbE\left[\bra J,W^{\eta}\left[\phi^{(\eta)}_{T/\eta}\right]\ket\right]=\bra J,\mu_T\ket,
\end{equation}
where the measure $\mu_T$ is given as the weak solution of the linear Boltzmann equation
\begin{equation}
\label{boltzmann_eq}
\frac{\partial}{\partial T}\mu_T(X,V)=-\sin{(2\pi V)}\cdot\nabla_X\mu(X,V)+\int_{\tthree}\dd U\sigma(U,V)\left(\mu_T(X,U)-\mu_T(X,V)\right).
\end{equation}
In (\ref{boltzmann_eq}), $\sin(2\pi V)$ is a vector with components $\sin(2\pi V_j)$, $j=1,2,3$ and the collision kernel $\sigma$ equals
\begin{equation}
\sigma(U,V)=2\pi\delta\left(e(U)-e(V)\right).
\end{equation}
\end{thm}
\begin{proof}
The corresponding Theorem 2.1 in \cite{chen_rs} is stated only for WKB initial data of the form
\begin{equation}
\label{chen_init}
\phi^{(\eta)}_0(x)=\eta^{3/2}h(\eta x)e^{iS(\eta x)/\eta}
\end{equation}
with $h, S\in\caS(\bbR^3)$ and a bounded random potential $V_\omega$. Our more general choice of initial conditions does not affect the proof at all, one just has to plug in (\ref{initialwignerconv}) whenever in \cite{chen_rs} the specific convergence behavior of (\ref{chen_init}) is used. The issue of the (very mildly) unbounded Gaussian $V_\omega$ can be either addressed as in \cite{chen_lp} by restricting the model to a box $\lbrace-L, -L+1,..., L\rbrace^3$ for very large $L$, or, as outlined in Section \ref{graph_exp} below, by introducing a cutoff version $V^L_\omega$ of the potential as an intermediate step to justify the Duhamel expansion.
\end{proof}
Having understood the kinetic limit for the disorder-averaged Wigner transform, we now focus on the issue of the deviations from this average. Chen \cite{chen_lp} proves that the $r$-th moment of the random variable $\bra J,W^{\eta}\left[\phi^{(\eta)}_{T/\eta}\right]\ket-\bra J,\mu_T\ket$ converges to $0$ for all $r>0$, in particular the variance of $\bra J,W^{\eta}\left[\phi^{(\eta)}_{T/\eta}\right]\ket$, corresponding to $r=2$. In this sense, the scaled Wigner transform becomes deterministic in the kinetic limit. However, due to reasons that will be explained in Section \ref{graph_estimates}, additional \emph{concentration of singularity} assumptions regarding the initial condition $\phi^{(\eta)}_0$ had to be made (see equations (29) - (31) of \cite{chen_lp}). Therefore, convergence of $r$-th moments could only be established for $\phi^{(\eta)}_0$ allowing for a decomposition
\begin{equation}
\label{chens_assumptions1}
\widehat{\phi^{(\eta)}_0}(k)=f^{(\eta)}_{\infty}(k)+f^{(\eta)}_{sing}(k)
\end{equation}
such that
\begin{equation}
\label{chens_assumptions2}
\left\Vert f^{(\eta)}_{\infty}\right\Vert_{L^\infty(\tthree)}\leq c
\end{equation}
and
\begin{equation}
\label{chens_assumptions3}
\left\Vert\left|f^{(\eta)}_{sing}\right|\ast\left|f^{(\eta)}_{sing}\right|\right\Vert_{L^2(\tthree)}\leq c'\eta^{4/5}
\end{equation}
for $c,c'$ constants independent of $\eta$. Our main result states that such extra assumptions are superfluous.

\begin{thm}
\label{theorem_main}
Let $\left(\phi^{(\eta)}_0\right)_{\eta>0}$ be a sequence of initial states such that (\ref{boundedprob}) holds. Then there is a constant $c>0$ such that for all macroscopic times $T>0$, all test functions $J\in\caS\left(\bbR^3\times\bbT^3\right)$, and $r\geq1$
\begin{equation}
\label{fluct_tozero}
\left(\bbE\left[\left|\bra J,W^{\eta}\left[\phi^{(\eta)}_{T/\eta}\right]\ket-\bbE\bra J,W^{\eta}\left[\phi^{(\eta)}_{T/\eta}\right]\ket\right|^r\right]\right)\leq C(J,T)^r \lambda^{c}
\end{equation}
for sufficiently small $\lambda>0$ and $C(J,T)<\infty$ only depending on $J,T$. Together with (\ref{initialwignerconv}) and Theorem \ref{theorem_convergence_exp} this yields 
\begin{equation}
\lim_{\lambda\rightarrow0}\bbE\left[\left|\bra J,W^{\eta}\left[\phi^{(\eta)}_{T/\eta}\right]\ket-\bra J,\mu_T\ket\right|^r\right]=0
\end{equation}
for all $r,T>0$, and consequently
\begin{equation}
\bra J,W^{\eta}\left[\phi^{(\eta)}_{T/\eta}\right]\ket\stackrel{\eta\rightarrow0}{\longrightarrow}\bra J,\mu_T\ket
\end{equation}
in probability.
\end{thm}
The next three sections provide a proof of the main theorem. Relying on notation already introduced in \cite{erdyau,chen_rs}, in Section \ref{graph_exp} we present the perturbative graph expansion of the time evolution. Section \ref{graph_estimates}  contains the main result: While in \cite{chen_lp} only the contribution of a certain subclass of graphs is controlled, we now provide a complete classification of graphs and sufficiently sharp estimates for all cases, no longer requiring the assumptions (\ref{chens_assumptions1}-\ref{chens_assumptions3}). A key ingredient for the improved bounds is a novel resolvent estimate, which is derived in Appendix \ref{app_2res}. Finally, we combine all the estimates in Section \ref{proof_main} to complete the proof of Theorem \ref{theorem_main}.

A short discussion of Theorem \ref{theorem_main} and the underlying assumptions are in order. We choose the sequence of initial states $\phi^{(\eta)}_0$ as general as possible, but only analyze the case of i.i.d. Gaussian random variables, $\omega_x$. This simplifies matters as we will be able to restrict our graph expansion in Section \ref{graph_exp} to pairings only.

For the general case, for example $\omega_x$ i.i.d. with $\omega_x$ bounded, from the cumulant expansion one has higher order graphs, which have to be shown to vanish in the kinetic limit. For the average Wigner transform this has been accomplished for the discrete Schr\"{o}dinger equation \cite{chen_rs} and for the random wave equation \cite{luspo}.
We expect such a technique to carry over to the study of the higher moments without any complications. On the other hand, one might be interested in the continuum model, for which the analogue of Theorem \ref{theorem_convergence_exp} was shown in \cite{erdyau}. For a result corresponding to Theorem \ref{theorem_main} in the continuum case, one would have to implement the new graph classification literally. An estimate of the kind of Lemma \ref{lemma_2res} is easy to derive (at least in dimensions $d\geq3$), since the geometry of the energy level sets (spheres of constant $k^2$) is much simpler. However, we stick with the lattice model and the Gaussian random potential, as considered in \cite{chen_lp}, to facilitate a direct comparison.

With only the assumptions from Theorem \ref{theorem_convergence_exp} on the initial state one could lose mass at infinity, which is no longer seen by the kinetic equation, of course. To avoid such unphysical situations, one would have to impose some tightness condition on the large space scale, as
\begin{equation}
\label{tightspace}
\lim_{R\rightarrow\infty}\overline{\lim_{\eta\rightarrow0}}\sum_{|x|>R/\eta}\left|\phi^{(\eta)}_0(x)\right|^2=0.
\end{equation}
Then, on a subsequence of $\eta$ chosen such that $\mu_0$ exists,
\begin{equation}
\mu_0\left(\bbR^3\times\tthree\right)=\lim_{\eta\rightarrow0}\left\Vert\phi^{(\eta)}_0\right\Vert^2_{\ell^2\left(\bbZ^3\right)}.
\end{equation}


As a consequence of Theorem \ref{theorem_main}, one can also establish an almost sure convergence result for $W^\eta\left[\mathrm{e}^{-iH_\omega\tau/\eta}\phi_0^{(\eta)}\right]$, for example for WKB initial states $\phi_0^{(\eta)}$.

\begin{thm}
\label{thm_almostsure}
Fix a sequence of initial states $\left(\phi_0^{\eta}\right)_{\eta>0}$ of the form (\ref{chen_init}), with $h,S\in\caS\left(\bbR^3\right)$. Then almost surely, i.e. for almost all realizations of the random potential $V$,
\begin{equation}
\lim_{\eta\rightarrow0}\sup_{\tau\in[0,T]}\left|\bra J,W^\eta\left[\mathrm{e}^{-iH_\omega\tau/\eta}\phi_0^{(\eta)}\right]\ket-\bra J,\mu_\tau\ket\right|=0
\end{equation}
for all $T\geq0$ and $J\in\caS\left(\bbR^3\times\bbT^3\right)$.
\end{thm}
The proof of this theorem will be given in Section \ref{almostsure}. It relies on a large deviation principle for the Gaussian random potential; however, a comparable result should also hold for other distributions of the one-site potential $V(x)$ as long one has a good control on the tails.

The assumption that the initial states $\left(\phi_0^{\eta}\right)_{\eta>0}$ form a WKB sequence is quite stringent compared to the conditions (\ref{boundedprob}) and (\ref{initialwignerconv}) of the previous theorems. While any other sequence of initial states that fulfills an estimate of the form (\ref{phi_cts}) would do equally well instead of WKB states, (\ref{phi_cts}) or similar estimates are necessary as one has to control the $\phi^{\eta}_0$-dependence of the set of ``bad" realizations of the potential $V$, i.e. the set of $\omega$ which lead to large values of 
\begin{equation}
\left|\bra J,W^\eta\left[\mathrm{e}^{-iH_\omega\tau/\eta}\phi_0^{(\eta)}\right]\ket-\bra J,\mu_\tau\ket\right|.
\end{equation}
Therefore, the only way to guarantee almost sure convergence is to essentially show it on a subsequence $\eta_n\rightarrow0$ by a Borel-Cantelli-argument and to ``tie" $\phi_0^\eta$ to $\phi_0^{\eta_n}$ by (\ref{phi_cts}). 

\section{Graph expansion}
\label{graph_exp}
From now on we only consider $\widehat{\phi^{(\eta)}_0}\in C^{\infty}\left(\tthree\right)$. As $e(k)$ is a smooth function of $k\in\tthree$, the unperturbed time evolution ${e}^{-itH_0}$ leaves this space invariant. As $V_\omega$ is an unbounded potential almost surely, we first have to make sure that the Duhamel expansion is well-defined. To achieve this, use the (very coarse) estimate
\begin{equation}
\left|V_\omega(x)\right|\leq Y(\omega)\left(1+|x|\right)
\end{equation}
with a random variable $Y$ with all moments bounded, and note that that $\phi^{(\eta)}_0$ with $\widehat{\phi^{(\eta)}_0}\in C^{\infty}\left(\tthree\right)$ decay faster than any negative power of $|x|$ in position space. Therefore, expressions like
\begin{equation}
\label{defphin}
\phi^{(\eta)}_{n,t}=(-i\lambda)^n\int_{\bbR_+^{n+1}}\dd s_0...\dd s_n\delta\left(\sum_{j=0}^ns_j-t\right)\prod_{j=0}^{n-1}\left(\mathrm{e}^{-is_jH_0}V_\omega\right)\mathrm{e}^{-is_nH_0}\phi^{(\eta)}_0
\end{equation}
are defined for all $n\in\bbN_0$. From now on suppressing the $\eta$ dependence, we Duhamel expand the perturbed time evolution
\begin{equation}
\phi_t=\phi\upm_t+R_{N,t}
\end{equation}
with
\begin{equation}
\label{defphimain}
\phi\upm_{t}=\sum_{n=0}^N\phi_{n,t}
\end{equation}
and a remainder $R_{N,t}$, with a large $N$ that will be optimized in section \ref{proof_main}. We first focus on the main part and need to understand the variance of the Wigner transform of $\phi\upm$, tested against a $J\in\caS\left(\bbR^3\times\tthree\right)$
\begin{equation}
\label{sum_covari}
\begin{split}
&\vari{\bra J,W^{\eta}\left[\phi\upm_{t}\right]\ket}\\
&\hspace{1cm}=\sum_{n_1...n_4=0}^N\covari{\bra J,W^{\eta}\left[\phi_{n_1,t},\phi_{n_2,t}\right]\ket,\bra J,W^{\eta}\left[\phi_{n_3,t},\phi_{n_4,t}\right]\ket}\\
&\hspace{1cm}\leq(N+1)^2\sum_{n_1,n_2=0}^N\vari{\bra J,W^{\eta}\left[\phi_{n_1,t},\phi_{n_2,t}\right]\ket}.
\end{split}
\end{equation}
Written out in momentum space, a single summand reads
\begin{equation}
\label{sing_sum_var}
\begin{split}
&\vari{\bra J,W^{\eta}\left[\phi_{n_1,t},\phi_{n_2,t}\right]\ket}\\
&\hspace{1cm}=\bbE\left[\int_{\bbR^{3}\times\tthree}\dd\xi\upo \dd v\upo\widehat{J_\eta}\left(\xi\upo,v\upo\right)\hat\phi_{n_1,t}\left(v\upo-\frac{\xi\upo}{2}\right)\overline{\hat\phi_{n_2,t}\left(v\upo+\frac{\xi\upo}{2}\right)}\right.\\
&\hspace{2cm}\left.\times\int_{\bbR^{3}\times\tthree}\dd\xi\upt \dd v\upt\overline{\widehat{J_\eta}\left(\xi\upt,v\upt\right)}\overline{\hat\phi_{n_1,t}\left(v\upt-\frac{\xi\upt}{2}\right)}\hat\phi_{n_2,t}\left(v\upt+\frac{\xi\upt}{2}\right)\right]\\
&\hspace{1.2cm}-\bbE\left[\int_{\bbR^{3}\times\tthree}\dd\xi\upo \dd v\upo\widehat{J_\eta}\left(\xi\upo,v\upo\right)\hat\phi_{n_1,t}\left(v\upo-\frac{\xi\upo}{2}\right)\overline{\hat\phi_{n_2,t}\left(v\upo+\frac{\xi\upo}{2}\right)}\right]\\
&\hspace{2cm}\times\bbE\left[\int_{\bbR^{3}\times\tthree}\dd\xi\upt \dd v\upt\overline{\widehat{J_\eta}\left(\xi\upt,v\upt\right)}\overline{\hat\phi_{n_1,t}\left(v\upt-\frac{\xi\upt}{2}\right)}\hat\phi_{n_2,t}\left(v\upt+\frac{\xi\upt}{2}\right)\right]
\end{split}
\end{equation}
The key part of the proof is to appropriately estimate (\ref{sing_sum_var}). To this end, one has to understand how $V_{\omega}$ act on Fourier space. We do not directly Fourier transform $V_\omega$, which is not summable, but the cut-off version
\begin{equation}
\label{vcutoff}
V_\omega^{L}(x)=V_\omega(x)1_{\lbrace|x|\leq L\rbrace}
\end{equation}
for large but finite $L$. Now on the one hand, it is straightforward to see that if we define ${\phi^{(\eta)}_{L,n,t}}$ by (\ref{defphin}) with $V_\omega$ replaced by $V_\omega^L$, we have
\begin{equation}
\label{conv_e_l}
\bbE\left[\left(\mathrm{dist}_{C^\infty}\left(\widehat{\phi^{(\eta)}_{L,n,t}},\widehat{\phi^{(\eta)}_{n,t}}\right)\right)^\gamma\right]\rightarrow0\hspace{5mm}(L\rightarrow\infty)
\end{equation}
with any $\gamma,t,\eta>0$ and $n\in\bbN_0$ fixed. On the other hand, for finite $L$, $\widehat{V_\omega^L}$ is well-defined, and expectations of the type $\bbE\left[\prod_{l=1}^{2s}\widehat{V^L}(k_l)\right]$ ($k_l\in\tthree$) are equivalent to sums over pairings, with each pairing contributing the product of $s$ factors of the type
\begin{equation}
\label{vpairl}
\bbE\left[\widehat{V^L}(k)\widehat{V^L}(p)\right]=\delta_L(k+p).
\end{equation}
As $L\rightarrow\infty$, (\ref{vpairl}) converges to $\delta(k+p)$ if tested against smooth functions of $k,p$.

As in \cite{chen_lp} one can represent $\widehat{\phi^{(\eta)}_{L,n,t}}$ by the resolvent formula (for a compact derivation, see \cite{luspo})
\begin{equation}
\label{res_expansion}
\begin{split}
\widehat{\phi^{(\eta)}_{L,n,t}}\left(k_0\right)=\frac{(-\lambda)^ne^{\eps t}}{2\pi i}&\int_{I}\dd \alpha e^{-i\alpha t}\\&\int_{\left(\tthree\right)^n}\dd k_1 ...\dd k_n\prod_{j=0}^n\frac{1}{e\left(k_j\right)-\alpha-i\eps}\prod_{j=1}^n \widehat{V^L}\left(k_{j-1}-k_j\right)\widehat{\phi^{(\eta)}_{0}}\left(k_n\right).
\end{split}
\end{equation}
for any $0<\eps<1$, with the integral $\int_{I}\dd \alpha$ performed clockwise along the edges of the rectangle with the vertices $\lbrace-1,7,7-i,-1-i\rbrace$ 
encircling $\mathrm{spec}\left(H_0-i\eps\right)=[0,6]-i\eps$. The pointwise complex conjugate of $I$ is denoted $\overline{I}$, and integration along $\overline{I}$ will always be performed counterclockwise. Now we adopt notation from \cite{chen_lp} which however will considerably simplify as, for the time being, we only consider the variance instead of higher momenta. For $j=1,2$ we set $\eps_j=(-1)^j\eps$ and $\overline{n}=n_1+n_2$ to define
\begin{equation}
\begin{split}
k^{(j)}&=\left(k_0^{(j)},...,k_{\overline{n}+1}^{(j)}\right)\\
\dd k^{(j)}&=\prod_{l=0}^{\overline{n}+1}\dd k_{l}^{(j)}\\
K^{(j)}\left(k^{(j)},\alpha_j,\beta_j,\eps\right)&=\prod_{l=0}^{n_1}\frac{1}{e(k^{(j)}_l)-\alpha_j+i\eps_j}\prod_{l^{'}=n_1+1}^{\oln+1}\frac{1}{e(k^{(j)}_{l^{'}})-\beta_j-i\eps_j}\\
U_L\left[k^{(j)}\right]&=\prod_{l=1}^{n_1}\widehat{V^L_\omega}\left(k^{(j)}_l-k^{(j)}_{l-1}\right)\prod_{l^{'}=n_1+1}^{\oln+1}\widehat{V^L_\omega}\left(k^{(j)}_{l^{'}}-k^{(j)}_{l^{'}-1}\right).
\end{split}
\end{equation}
Thus, in view of (\ref{conv_e_l}) and (\ref{res_expansion}) we arrive at
\begin{equation}
\label{var_sum_amp}
\begin{split}
(\ref{sing_sum_var})=&\lim_{L\rightarrow\infty}\vari{\bra J,W^{\eta}\left[\phi_{L,n_1,t},\phi_{L,n_2,t}\right]\ket}\\=&\lim_{L\rightarrow\infty}\frac{e^{4\eps t}\lambda^{2\oln}}{(2\pi)^4}\int_{I}\dd\alpha_1\int_{\overline{I}}\dd\beta_1\int_{\overline{I}}\dd\alpha_2\int_{I}\dd\beta_2e^{-it(\alpha_1-\beta_1-\alpha_2+\beta_2)}\\
&\int_{\left(\tthree\right)^{2\oln+4}}\dd k\upo\dd k\upt\left(\bbE\left[\overline{U_L\left[k^{(1)}\right]}U_L\left[k^{(2)}\right]\right]-\bbE\left[\overline{U_L\left[k^{(1)}\right]}\right]\bbE\left[U_L\left[k^{(2)}\right]\right]\right)\\
&\int_{\bbR^{3}}\dd\xi^{(1)}\int_{\bbR^{3}}\dd\xi^{(2)}\delta\left(k\upo_{n_1+1}-k\upo_{n_1}-\xi\upo\right)\delta\left(k\upt_{n_1+1}-k\upt_{n_1}-\xi\upt\right)\\
&\hspace{1cm}{\widehat{J_\eta}\left(\xi\upo,\left(k_{n_1}\upo+k_{n_1+1}\upo\right)/2\right)}\overline{\widehat{J_\eta}\left(\xi\upt,\left(k_{n_1}\upt+k_{n_1+1}\upt\right)/2\right)}\\
&\hspace{1cm}K^{(1)}\left(k^{(1)},\alpha_1,\beta_1,\eps\right)K^{(2)}\left(k^{(2)},\alpha_2,\beta_2,\eps\right)\\
&\hspace{1cm}\widehat{\phi_0}\left(k\upo_0\right)\overline{\widehat{\phi_0}\left(k\upo_{\oln+1}\right)}\overline{\widehat{\phi_0}\left(k\upt_0\right)}\widehat{\phi_0}\left(k\upt_{\oln+1}\right)\\
=&\sum_{\pi\in\Pi^{conn}_{n_1,n_2}}\ampl(\pi)
\end{split}
\end{equation}
with the \emph{amplitude} of each pairing $\pi$ given as
\begin{equation}
\label{def_amp}
\begin{split}
\ampl(\pi)=&\frac{e^{4\eps t}\lambda^{2\oln}}{(2\pi)^4}\int_{I}\dd\alpha_1\int_{\overline{I}}\dd\beta_1\int_{\overline{I}}\dd\alpha_2\int_{I}\dd\beta_2e^{-it(\alpha_1-\beta_1-\alpha_2+\beta_2)}\\
&\int_{\left(\tthree\right)^{2\oln+4}}\dd k\upo\dd k\upt\delta_\pi\left(k\upo{},k\upt\right)\\
&\int_{\bbR^{3}}\dd\xi^{(1)}\int_{\bbR^{3}}\dd\xi^{(2)}\delta\left(k\upo_{n_1+1}-k\upo_{n_1}-\xi\upo\right)\delta\left(k\upt_{n_1+1}-k\upt_{n_1}-\xi\upt\right)\\
&\hspace{1cm}{\widehat{J_\eta}\left(\xi\upo,\left(k_{n_1}\upo+k_{n_1+1}\upo\right)/2\right)}\overline{\widehat{J_\eta}\left(\xi\upt,\left(k_{n_1}\upt+k_{n_1+1}\upt\right)/2\right)}\\
&\hspace{1cm}K^{(1)}\left(k^{(1)},\alpha_1,\beta_1,\eps\right)K^{(2)}\left(k^{(2)},\alpha_2,\beta_2,\eps\right)\\
&\hspace{1cm}\widehat{\phi_0}\left(k\upo_0\right)\overline{\widehat{\phi_0}\left(k\upo_{\oln+1}\right)}\overline{\widehat{\phi_0}\left(k\upt_0\right)}\widehat{\phi_0}\left(k\upt_{\oln+1}\right).\\
\end{split}
\end{equation}

In (\ref{var_sum_amp}) and (\ref{def_amp}) the $\Pi^{conn}_{n_1,n_2}$ and $\delta_\pi$ notations are explained as follows. In keeping with \cite{chen_lp} we visualize the scattering process leading to (\ref{var_sum_amp}) by writing each resolvent as a solid \emph{propagator line} and each interaction with a potential as a black bullet to arrive at a graph like Figure \ref{fig_nopairing}. The two parallel components of Figure \ref{fig_nopairing} are called the first and second \emph{one-particle line}, respectively. Taking the expectation $\bbE\left[\overline{U_L\left[k^{(1)}\right]}U_L\left[k^{(2)}\right]\right]$ will yield a sum over all pairings of the $2\oln$ dark bullets, with each pairing $\pi$ contributing a product $\delta_\pi$ of $\oln$ delta distributions in the $L\rightarrow\infty$ limit, each delta either of the form
\begin{equation}
\label{def_internal}
\delta\left(k^{(j)}_{l+1}-k^{(j)}_{l}+k^{(j)}_{l^{'}+1}-k^{(j)}_{l^{'}}\right),\hspace{5mm}|l-l^{'}|\geq1,\hspace{5mm}j\in{1,2}
\end{equation}
called an \emph{internal delta}, or  
\begin{equation}
\label{def_transfer}
\delta\left(k^{(1)}_{l+1}-k^{(1)}_{l}-k^{(2)}_{l^{'}+1}+k^{(2)}_{l^{'}}\right).
\end{equation}
a \emph{transfer delta}. Note the different signs in (\ref{def_internal}) and (\ref{def_transfer}), which are due to $\widehat{V^L_\omega}(k)=\overline{\widehat{V^L_\omega}(-k)}$. The product $\bbE\left[\overline{U_L\left[k^{(1)}\right]}\right]\bbE\left[U_L\left[k^{(2)}\right]\right]$, on the other hand, will only produce internal deltas in the $L\rightarrow\infty$ limit, so taking the difference will leave us with the $\delta_\pi$ produced by $\pi\in\Pi^{conn}_{n_1,n_2}$, the set of all pairings that comprise at least one transfer delta, i.e. one contraction between the first and second one-particle line. We will analyze the contribution of $\ampl(\pi)$ for each $\pi\in\bigpi$ in section \ref{graph_estimates}, where we will also give examples of such pairings will be given in Figures \ref{fig_non_simple} to \ref{fig_anti}.

For the remainder term, on the other hand, Lemma 3.14 in \cite{chen_rs} uses the method of partial time integration introduced in \cite{erdyau} to derive
\begin{equation}
\label{est_remainder}
\begin{split}
\bbE&\left[\left\Vert R_{N,t}\right\Vert_{\ell^2}^2\right]\\&\leq\Vert\phi_0\Vert^2_{\ell^2}\left[\frac{N^2\kappa^2(C\lambda^2\eps^{-1})^{4N}}{(N!)^{1/2}}\right.\\&+N^2\kappa^2(C\lambda^2\eps^{-1}|\log\eps|)^{4N}|\log\eps|^3\left(\eps^{1/5}(4N)!+\eps^2(4N)^{20N}\right)\\&+\eps^{-2}(C\lambda^2\eps^{-1}|\log\eps|)^{4N}|\log\eps|^3\left(\kappa^{-N}(4N)!+\kappa^{-N+5}\eps(4N)!(4N)^4\right.\\&\hspace{5.5cm}\left.\left.+\kappa^{-N+9}\eps^2(4N)!(4N)^8+\eps^3(4N)^{20N}\right)\vphantom{\frac{N^2\kappa^2(C\lambda^2\eps^{-1})^{4N}}{(N!)^{1/2}}}\right],
\end{split}
\end{equation}
for all $N,\kappa\in\bbN$ and $\eps\leq t^{-1}$, with a constant $C$ that does not depend on any parameter of the problem. The graph expansion method used in the proof of (\ref{est_remainder}) carries over to Gaussian instead of bounded $V_\omega$ without complications by the above cut-off argument. $\kappa, \eps$ will as well be fixed in section \ref{proof_main}. At this point it is in order to remark that the estimate of the remainder term is a very important idea that was crucial for \cite{erdyau,chen_rs,chen_lp}, and also for the paper at hand. The new ideas needed to verify our main theorem, however, are only improved estimates for the main term, so we just quote (\ref{est_remainder}) and concentrate on $\phi\upm$ from now on.

\begin{figure}[htb] 
\centering 
\def\svgwidth{400pt} 
\input{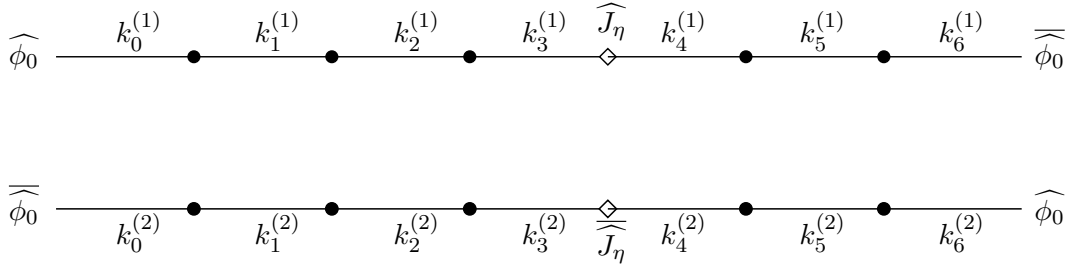} 
\caption{A graph with $n_1=3$, $n_2=2$, $\oln=5$.}
\label{fig_nopairing} 
\end{figure} 
\section{Graph estimates}
To prove that the contribution of connected graphs and thus the variance vanishes in the $\eps\rightarrow0$ limit, the key ingredient in \cite{chen_lp} is the three-resolvent integral, 
\begin{equation}
\label{three_res_int}
\sup_{\gamma_i\in I}\sup_{k\in\tthree}\int_{\left(\tthree\right)^2}\frac{\dd p\dd q}{\left|e(p)-\gamma_1-i\eps\right|\left|e(q)-\gamma_2-i\eps\right|\left|e(p\pm q+k)-\gamma_3-i\eps\right|}\leq\eps^{-4/5}\left|\log\eps\right|^3
\end{equation}
which originally was proven in \cite{chen_rs}, Lemma 3.11 to control the contribution of ``crossing pairings". To be able to apply (\ref{three_res_int}) to the present problem one has to make sure that the three resolvents in question are ``isolated" from the wave functions $\phi_0^{(\eta)}$ in the sense that the two integration variables of (\ref{three_res_int}) are independent of the arguments of the $\phi_0^{(\eta)}$. This is verified in \cite{chen_lp} for a subclass of graphs that (roughly) corresponds to those $\pi\in\bigpi$ that have a one-particle line with a ``generalized crossing" acording to the Definition below. If one still wants to make use of (\ref{three_res_int}) for $\pi$ not in this class, one will have to introduce ad hoc assumptions like (\ref{chens_assumptions1})-(\ref{chens_assumptions3}), controlling the overlap of singularities of the wave function and of the resolvents, thereby reducing the admissible set of initial states. However, noticing that it is possible to ``isolate" not necessarily three, but at least two resolvents in this situation, we can establish a two-resolvent-estimate, Lemma \ref{lemma_2res}, that will suffice to estimate the remaining cases. After reformulating the ideas of \cite{chen_lp} in sections \ref{basic_estimate} and \ref{gen_crossing}, we show how to control the graphs not analyzed so far in sections \ref{crossing_transfer} and \ref{para_anti}.
We start with a classification of pairings.  
\begin{defin}
\label{graph_class}
Each pairing $\pi\in\Pi^{conn}_{n_1,n_2}$ belongs exactly one of the following three cases.
\begin{itemize}
\item We say that a pairing has a \emph{generalized crossing} on the $j$-th one-particle line, if the $j$-th one-particle line contains a crossing of internal contractions, i.e. there are two deltas $\delta\left(k^{(j)}_{l_2+1}-k^{(j)}_{l_2}+k^{(j)}_{l_1+1}-k^{(j)}_{l_1}\right)$ and $\delta\left(k^{(j)}_{i_2+1}-k^{(j)}_{i_2}+k^{(j)}_{i_1+1}-k^{(j)}_{i_1}\right)$ with $l_1<i_1<l_2<i_2$ or there is a transfer delta $\delta\left(k^{(j)}_{i_1+1}-k^{(j)}_{i_1}-k^{(j^{'})}_{i^{'}_1+1}+k^{(j^{'})}_{i^{'}_1}\right)$ such that its end on the $j$-th one-particle line lies between the endpoints of an internal delta $\delta\left(k^{(j)}_{l_2+1}-k^{(j)}_{l_2}+k^{(j)}_{l_1+1}-k^{(j)}_{l_1}\right)$ with $l_1<i_1<l_2$. Examples are shown in Figures \ref{fig_non_simple}, \ref{fig_prot_gate} and \ref{fig_prot_rung}.
\item We write that a pairing with no generalized crossing on either of the one-particle lines has \emph{parallel (or anti-parallel)} transfer contractions if there are a total of $m$ transfer deltas, $\delta\left(k^{(1)}_{l_1+1}-k^{(1)}_{l_1}-k^{(2)}_{l_1^{'}+1}+k^{(2)}_{l_1^{'}}\right)$ through\\ $\delta\left(k^{(1)}_{l_m+1}-k^{(1)}_{l_m}-k^{(2)}_{l_m^{'}+1}+k^{(2)}_{l_m^{'}}\right)$, and the sequences $l$, $l^{'}$ are both increasing (or $l$ is increasing, while $l^{'}$ is decreasing). Examples are shown in Figures \ref{fig_parallel} and \ref{fig_anti}. A graph with only one transfer contraction (which does not lead to a generalized crossing), is consequently classified as both parallel and anti-parallel.
\item A pairing with no generalized crossing on its one-particle lines has \emph{crossing transfer contractions} if the transfer contractions are neither parallel nor anti-parallel, as shown in Figure \ref{fig_crossing}.
\end{itemize}
\end{defin}
\noindent{}\emph{Nested pairings} on a one-particle line, as defined in Definition 2.5 of \cite{erdyau} play no special role in our analysis, we always can make use of other structures and only mention them for completeness.

\label{graph_estimates}
\subsection{Basic estimate}
We start with a basic estimate that is not sharp enough to directly verify our convergence result, but instructive as the improved results below can be understood as slight modifications to the proof of the estimates in this section. As in \cite{chen_lp}, we introduce new integration variables, the \emph{transfer momenta} $u_1,...,u_m\in\tthree$ for all $m$ transfer contractions, thus rewriting every transfer delta
\begin{equation}
\delta\left(k^{(1)}_{l_r+1}-k^{(1)}_{l_r}-k^{(2)}_{l_r^{'}+1}+k^{(2)}_{l_r^{'}}\right)\rightarrow\delta\left(k^{(1)}_{l_r+1}-k^{(1)}_{l_r}-u_r\right)\delta\left(k^{(2)}_{l_r^{'}+1}+k^{(2)}_{l_r^{'}}-u_r\right),
\end{equation}
$r=1,...,m$. For $j=1,2$, and $u=(u_1,...,u_m)$ we can now define the internal deltas on the $j$-th one-particle line
\begin{equation}
\delta_{int}^{(j)}\left(k^{(j)}\right)=\prod_{\mbox{internal deltas}}\delta\left(k^{(j)}_{l+1}-k^{(j)}_{l}+k^{(j)}_{l^{'}+1}-k^{(j)}_{l^{'}}\right)
\end{equation}
and
\begin{equation}
\delta^{(j)}\left(u,k^{(j)}\right)=\delta_{int}^{(j)}\left(k^{(j)}\right)\prod_{r=1}^m\delta\left(k^{(j)}_{l_r+1}-k^{(j)}_{l_r}-u_r\right)
\end{equation}
to arrive at what is Definition 5.3 and Lemma 5.3 in \cite{chen_lp}.
\begin{lma}
\label{lem_factorization}
For a $\pi\in\Pi^{conn}_{n_1,n_2}$ with $m$ transfer deltas we let $\pi_j$ define the the graph defined by $\delta_{int}^{(j)}$ on the $j$-th one-particle line, and have
\begin{equation}
\label{factorization}
\ampl(\pi)=\int_{(\tthree)^m}\dd u_1...\dd u_m \ampl_1\left(u,\pi_1\right)\ampl_2\left(u,\pi_2\right)
\end{equation}
with
\begin{equation}
\label{def_amp1}
\begin{split}
\ampl_1\left(u,\pi_1\right)=&\frac{e^{2\eps t}\lambda^{\oln}}{(2\pi)^2}\int_{I}\dd\alpha_1\int_{\overline{I}}\dd\beta_1e^{-it(\alpha_1-\beta_1)}\\
&\int_{\left(\tthree\right)^{\oln+2}}\dd k\upo\delta^{(1)}\left(u,k^{(1)}\right)\int_{\bbR^{3}}\dd\xi^{(1)}\delta\left(k\upo_{n_1+1}-k\upo_{n_1}-\xi\upo\right)\\
&\hspace{1cm}{\widehat{J_\eta}\left(\xi\upo,\left(k_{n_1}\upo+k_{n_1+1}\upo\right)/2\right)}K^{(1)}\left(k^{(1)},\alpha_1,\beta_1,\eps\right)\\
&\hspace{1cm}\widehat{\phi_0}\left(k\upo_0\right)\overline{\widehat{\phi_0}\left(k\upo_{\oln+1}\right)}\\
\end{split}
\end{equation}
and
\begin{equation}
\label{def_amp2}
\begin{split}
\ampl_2\left(u,\pi_2\right)=&\frac{e^{2\eps t}\lambda^{\oln}}{(2\pi)^2}\int_{\overline{I}}\dd\alpha_2\int_{I}\dd\beta_2e^{it(\alpha_2-\beta_2)}\\
&\int_{\left(\tthree\right)^{\oln+2}}\dd k\upt\delta^{(2)}\left(u,k^{(2)}\right)\int_{\bbR^{3}}\dd\xi^{(2)}\delta\left(k\upt_{n_1+1}-k\upt_{n_1}-\xi\upt\right)\\
&\hspace{1cm}\overline{{\widehat{J_\eta}\left(\xi\upt,\left(k_{n_1}\upt+k_{n_1+1}\upt\right)/2\right)}}K^{(2)}\left(k^{(2)},\alpha_2,\beta_2,\eps\right)\\
&\hspace{1cm}\overline{\widehat{\phi_0}\left(k\upt_0\right)}{\widehat{\phi_0}\left(k\upt_{\oln+1}\right)}.\\
\end{split}
\end{equation}
\end{lma}
For each of the factors $\ampl_1$ and $\ampl_2$ in (\ref{factorization}), there is the following estimate
\begin{lma}
\label{lem_factor}
With $j=1,2$, $m$ the number of transfer contractions,$0\leq a \leq m$, and the $m$ $\tthree$ components of $u$ partitioned into $v\in\left(\tthree\right)^a$ and $w\in\left(\tthree\right)^{m-a}$, we have the estimate
\begin{equation}
\sup_w\int_{\left(\tthree\right)^a}\dd v\left|\ampl_j\left(u,\pi_j\right)\right|\leq e^{2\eps t}\lambda^{\oln}\eps^{a-\frac{\oln+m}{2}}\left|c\log\eps\right|^{a+\frac{\oln-m}{2}+2}\left\Vert\phi_0\right\Vert_{\ell^2}^2.
\end{equation}
for all $0<\eps\leq1/3$, with a constant $c$ depending only on the test function $J$.
\end{lma}
\begin{proof}
This is essentially just an extension of an argument used for the basic estimates in the proof of Theorem \ref{theorem_convergence_exp}, one only has to pay attention to the new contributions coming from the transfer deltas. A detailed proof is given in \cite{chen_lp}, we give a short sketch of the idea. Without loss of generality, we consider $j=1$, estimate
\begin{equation}
\label{estimate_opl}
\begin{split}
\left|\ampl_1\left(u,\pi_1\right)\right|\leq&\frac{e^{2\eps t}\lambda^{\oln}}{(2\pi)^2}\int_{I}|\dd\alpha_1|\int_{\overline{I}}|\dd\beta_1|\\
&\int_{\left(\tthree\right)^{\oln+2}}\dd k\upo\delta^{(1)}\left(u,k^{(1)}\right)\int_{\bbR^{3}}\dd\xi^{(1)}\delta\left(k\upo_{n_1+1}-k\upo_{n_1}-\xi\upo\right)\\
&\hspace{1cm}\left|{\widehat{J_\eta}\left(\xi\upo,\left(k_{n_1}\upo+k_{n_1+1}\upo\right)/2\right)}\right|\left|K^{(1)}\left(k^{(1)},\alpha_1,\beta_1,\eps\right)\right|\\
&\hspace{1cm}\left(\frac{1}{2}\left|\widehat{\phi_0}\left(k\upo_0\right)\right|^2+\frac{1}{2}\left|{\widehat{\phi_0}\left(k\upo_{\oln+1}\right)}\right|^2\right)\\
\end{split}
\end{equation}
and concentrate on the summand containing $\left|\widehat{\phi_0}\left(k\upo_0\right)\right|^2$. Here, and in all proofs of the lemmas below, it is important to note that not all (only roughly half) of the $k^{(j)}_l$ variables are free integration variables, as the others are are then fixed by the delta functions contained in $\delta_\pi$. Before evaluating integrals like (\ref{estimate_opl}), we will therefore have to choose our "free" and "dependent" variables consistently with the $\delta_\pi$ at hand. For the present proof, all momenta $k\upo_l$ except for $k\upo_0$ and $k\upo_{n_1+2}$ are classified either as \emph{free} if there is an internal delta $\delta\left(k\upo_l-k\upo_{l-1}+k\upo_i-k\upo_{i-1}\right)$ with $l<i$ or there is a transfer delta $\delta\left(k\upo_l-k\upo_{l-1}-u_i\right)$ with $u_i$ contained in $v$ (there are $(\oln-m)/2+a$ such indices) or \emph{dependent} if $l$ belongs to the other $(\oln+m)/2-a$ indices. Now we take the trivial $L^\infty$ estimate $\frac{1}{\eps}$ of all resolvents belonging to dependent momenta, and then integrate out all free momenta from $k\upo_{\oln+1}$ to $k\upo_{n_1+2}$ with the $L^1$ estimate (\ref{intest3d}). Next, we perform the $\beta_1$ integral, yielding another $\left|\log\eps\right|$ from the $k\upo_{n_1+1}$ resolvent. $k\upo_{n_1+1}$ is fixed by the delta $\delta\left(k\upo_{n_1+1}-k\upo_{n_1}-\xi\upo\right)$. Then we continue with $L^1$ estimates for all free momenta from $k\upo_{n_1}$ to $k\upo_{1}$, integrate over $\alpha_1$ for another $\left|\log\eps\right|$ from the $k\upo_{0}$ resolvent and get a factor $\left\Vert\phi_0\right\Vert_{\ell^2}^2$ from the $k\upo_0$ integral. Finally, the $\xi\upo$ integral contributes a $\eta$-indpendent factor $\int\dd\xi\sup_{p}\left|\widehat{J}(\xi,p)\right|$. For the summand with $\left|\widehat{\phi_0}\left(k\upo_{\oln+1}\right)\right|^2$, the order of integration and the definition of free and dependent would just be reversed, with left and right interchanged. Collecting all factors proves the lemma.
\end{proof}
\begin{cor}
\label{cor_basic}
For all $\pi\in\Pi^{conn}_{n_1,n_2}$,
\begin{equation}
\left|\ampl(\pi)\right|\leq e^{4\eps t}\lambda^{2\oln}\eps^{-\oln}\left|c\log\eps\right|^{\oln+4}\left\Vert\phi_0\right\Vert_{\ell^2}^4,
\end{equation}
for all $0<\eps\leq1/3$, with a constant $c$ depending only on $J$.
\end{cor}
\begin{proof}
With $j=1$, $j^{'}=2$ (or vice versa) use Lemmas \ref{lem_factorization} and \ref{lem_factor} to obtain
\begin{equation}
\begin{split}
\left|\ampl(\pi)\right|&\leq\int_{\left(\tthree\right)^m}\dd u \left|\ampl_j\left(u,\pi_j\right)\right|\sup_{u^{'}\in\left(\tthree\right)^m}\left|\ampl_{j^{'}}\left(u^{'},\pi_{j^{'}}\right)\right|\\
&\leq e^{4\eps t}\lambda^{2\oln}\eps^{-\oln}\left|c\log\eps\right|^{\oln+4}\left\Vert\phi_0\right\Vert_{\ell^2}^4.
\end{split}
\end{equation}
\end{proof}
\label{basic_estimate}
\subsection{Generalized crossings on the one-particle line}
\begin{lma}
\label{lem_gen_crossing}
Let $\pi\in\bigpi$ have a generalized crossing on the $j$-th one-particle line. Then, for all $0<\eps\leq1/3$, the estimate for $\ampl(\pi)$ of Corollary \ref{cor_basic} can be improved to
\begin{equation}
\left|\ampl(\pi)\right|\leq e^{4\eps t}\lambda^{2\oln}\eps^{\frac{1}{5}-\oln}\left|c\log\eps\right|^{\oln+5}\left\Vert\phi_0\right\Vert_{\ell^2}^4.
\end{equation}
\end{lma}
\begin{proof}
Without loss of generality, we can assume a generalized crossing on the first one-particle line and expand $\ampl_1(\pi)$ as in (\ref{estimate_opl}). We take $a=m$ (all transfer deltas contribute their free variable to the first one-particle line) and focus on the summand containing $\left|\widehat{\phi_0}\left(k\upo_0\right)\right|^2$. For deltas labeled as in Definition \ref{graph_class}, we follow \cite{luspo} and call $\left\lbrace i_1,...,l_2\right\rbrace$ the \emph{crossing interval}, and the generalized crossing is called \emph{minimal} if its crossing interval does not contain the crossing interval of another generalized crossing as a proper subset.  By a simple induction argument along the lines of \cite{luspo} we see that there is always a minimal generalized crossing, and label its indices $l_1, i_1, l_2$ again. We adopt the notions of free and dependent from above, first take the $L^\infty$ estimate of all dependent resolvents except for $k\upo_{l_2+1}$, and then integrate out all free $k\upo_{l}$ with $l\in\left\lbrace l_2+2,..., \oln+1\right\rbrace$, using the $L^1$ estimate for each of their resolvents. If $l_2<n_1$, we also take the $\beta_1$ integral and the $k\upo_{n_1+1}$ integrals with the same estimates as before. If not, the $\beta_1$ and $k\upo_{n_1+1}$ will be evaluated at some later point, whenever all $k\upo_l$ with $l>n_1+1$ have been taken care of. Now deviating from the order of integration in the proof of Lemma \ref{lem_factor}, we write
\begin{equation}
k\upo_{l_2+1}=k\upo_{l_2}+k\upo_{l_1}-k\upo_{l_1+1},
\end{equation}
and note that by construction the crossing interval is not connected to the outside by internal or transfer deltas at all, so
\begin{equation}
\tilde{\xi}\upo=k\upo_{l_2}-k\upo_{i_1+1}=\begin{cases}\xi\upo &\mbox{if } l_1<n_1<l_2\\ 0 &\mbox{otherwise,}\end{cases},
\end{equation}
and
\begin{equation}
k\upo_{l_2+1}=k\upo_{i_1+1}-k\upo_{l_1+1}+k\upo_{l_1}+\tilde{\xi}\upo.
\end{equation}
Before touching any of the other remaining variables, we evaluate the three-resolvent integral over the $k\upo_{l_1+1}$, $k\upo_{i_1+1}$ and $k\upo_{l_2+1}$ resolvents,
\begin{equation}
\label{gc_threeres}
\begin{split}
\sup_{\gamma_i\in I}&\int_{\left(\tthree\right)^2}\dd k\upo_{l_1+1} \dd k\upo_{i_1+1} \frac{1}{\left|e\left(k\upo_{l_1+1}\right)-\gamma_1-i\eps\right|\left|e\left(k\upo_{i_1+1}\right)-\gamma_2-i\eps\right|}\\&\hspace{1cm}\times\frac{1}{\left|e\left( k\upo_{i_1+1}-k\upo_{l_1+1}+k\upo_{l_1}+\tilde{\xi}\upo\right)-\gamma_3-i\eps\right|}\\&\hspace{0.5cm}\leq\eps^{-4/5}\left|\log\eps\right|^3
\end{split}
\end{equation}
with all $\gamma_i$ being equal to $\alpha_1$ or $\beta_1$. Then we take $L^1$ estimates for all remaining free momenta, and finally treat the integrals over $\alpha_1$, $k\upo_0$, and $\xi\upo$ as in the proof of Lemma \ref{lem_factor}. All together, we have collected the same factors as before, with three exceptions. Two of the momenta $k\upo_{l_1+1}$, $k\upo_{i_1+1}$ and $k\upo_{l_2+1}$ are free, one is dependent, so before they would have contributed a factor $\eps^{-1}|\log\eps|^2$, now we use (\ref{gc_threeres}) instead. Thus, after applying the same idea to the summand with the factor $\left|\widehat{\phi_0}\left(k\upo_{\oln+1}\right)\right|^2$ (with the necessary changes made to the definitions of "crossing interval" and integration order), the estimate for
\begin{equation}
\int_{(\tthree)^m}\dd u\left|\ampl_1\left(u,\pi_1\right)\right|
\end{equation}
improves by a factor $\eps^{1/5}|\log\eps|$. By plugging this improved bound into
\begin{equation}
\int_{(\tthree)^m}\dd u\left|\ampl_1\left(u,\pi_1\right)\right|\sup_{u^{'}}\left|\ampl_2\left(u^{'},\pi_2\right)\right|,
\end{equation}
this factor carries over to the estimate from Corollary \ref{cor_basic}.
\end{proof}

\begin{figure}[htb] 
\centering 
\def\svgwidth{400pt} 
\input{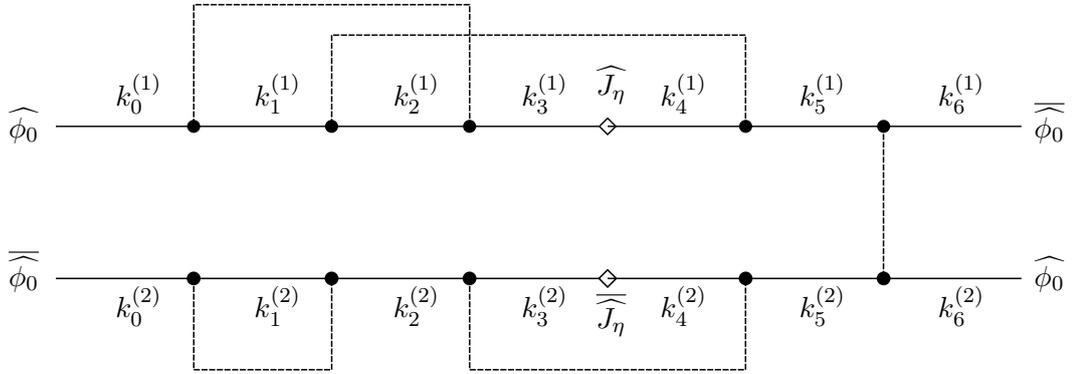} 
\caption{A graph with a vanishing contribution because of a crossing in the first one-particle line.}
\label{fig_non_simple} 
\end{figure} 

\begin{figure}[htb] 
\centering 
\def\svgwidth{400pt} 
\input{prot_gate_adjusted.pdf_tex} 
\caption{A generalized crossing on the second one-particle line, induced by the transfer contraction forcing $\delta\left(k\upo_3-k\upo_2-k\upt_2+k\upt_1\right)$.}
\label{fig_prot_gate} 
\end{figure} 

\begin{figure}[htb] 
\centering 
\def\svgwidth{400pt} 
\input{prot_rung_adjusted.pdf_tex} 
\caption{A generalized crossing on the second one-particle line, induced by transfer contraction forcing $\delta\left(k\upo_3-k\upo_2-k\upt_5+k\upt_4\right)$.}
\label{fig_prot_rung} 
\end{figure} 
\label{gen_crossing}
\subsection{Crossing transfer contractions}
For crossing transfer contractions, we essentially have the same estimates as for the case of generalized crossings on a one-particle line.
\begin{lma}
\label{lem_crossing_transfer}
Let $\pi\in\bigpi$ have no generalized crossings on either one-particle line, but at least one pair of crossing transfer contractions. Then the estimate for $\ampl(\pi)$ of Corollary \ref{cor_basic} can as well be improved to
\begin{equation}
\left|\ampl(\pi)\right|\leq e^{4\eps t}\lambda^{2\oln}\eps^{\frac{1}{5}-\oln}\left|c\log\eps\right|^{\oln+5}\left\Vert\phi_0\right\Vert_{\ell^2}^4
\end{equation}
for all $0<\eps\leq1/3$.
\end{lma}
\begin{proof}
For this proof, it is insufficient to split up the amplitude into $\ampl_1$ and $\ampl_2$ by Lemma \ref{lem_factorization}, and one rather has to consider the complete estimate
\begin{equation}
\begin{split}
\left|\ampl(\pi)\right|\leq&\frac{e^{4\eps t}\lambda^{2\oln}}{(2\pi)^4}\int_{I}\left|\dd\alpha_1\right|\int_{\overline{I}}\left|\dd\beta_1\right|\int_{\overline{I}}\left|\dd\alpha_2\right|\int_{I}\left|\dd\beta_2\right|\\
&\int_{\left(\tthree\right)^{2\oln+4}}\dd k\upo\dd k\upt\delta_\pi\left(k\upo{},k\upt\right)\\
&\int_{\bbR^{3}}\dd\xi^{(1)}\int_{\bbR^{3}}\dd\xi^{(2)}\delta\left(k\upo_{n_1+1}-k\upo_{n_1}-\xi\upo\right)\delta\left(k\upt_{n_1+1}-k\upt_{n_1}-\xi\upt\right)\\
&\hspace{5mm}\left|\widehat{J_\eta}\left(\xi\upo,\left(k_{n_1}\upo+k_{n_1+1}\upo\right)/2\right)\overline{\widehat{J_\eta}\left(\xi\upt,\left(k_{n_1}\upt+k_{n_1+1}\upt\right)/2\right)}\right|\\
&\hspace{5mm}\left|K^{(1)}\left(k^{(1)},\alpha_1,\beta_1,\eps\right)K^{(2)}\left(k^{(2)},\alpha_2,\beta_2,\eps\right)\right|\\
&\hspace{5mm}\left(\frac{1}{2}\left|\widehat{\phi_0}\left(k\upo_0\right)\right|^2+\frac{1}{2}\left|\widehat{\phi_0}\left(k\upo_{\oln+1}\right)\right|^2\right)\left(\frac{1}{2}\left|{\widehat{\phi_0}\left(k\upt_0\right)}\right|^2+\frac{1}{2}\left|\widehat{\phi_0}\left(k\upt_{\oln+1}\right)\right|^2\right).
\end{split}
\end{equation}
We first focus on the summand containing the factor $\left|\widehat{\phi_0}\left(k\upo_0\right)\right|^2\left|{\widehat{\phi_0}\left(k\upt_0\right)}\right|^2$ and call the $k^{(j)}_l$, $l\notin\lbrace0,n_1+1\rbrace$ with the same convention as before, \emph{free} if there is an internal delta $\delta\left(k^{(j)}_l-k^{(j)}_{l-1}+k^{(j)}_i-k^{(j)}_{i-1}\right)$, with $j=1,2$ and $l<i$, or if $j=1$ and there is a transfer delta $\delta\left(k^{(1)}_l-k^{(1)}_{l-1}-k^{(2)}_i+k^{(2)}_{i-1}\right)$, while we call it \emph{dependent} in all other cases. Furthermore, as the transfer contractions are not parallel, there is a pair of transfer deltas $\delta\left(k^{(1)}_{l_1}-k^{(1)}_{l_1-1}-k^{(2)}_{l^{'}_1}+k^{(2)}_{l^{'}_1-1}\right)$ and $\delta\left(k^{(1)}_{l_2}-k^{(1)}_{l_2-1}-k^{(2)}_{l^{'}_2}+k^{(2)}_{l^{'}_2-1}\right)$ such that $l_1<l_2$ but $l^{'}_1>l^{'}_2$. If several such pairs exist, we choose a particular set of indices ${l}_1$, ${l}^{'}_1$, ${l}_2$ and ${l}^{'}_2$ as follows. First, we take ${l}^{'}_2$ as small as possible, and then, after having fixed ${l}^{'}_2$ and thus $l_2$, we select the largest $l_1<l_2$ that does the job.
For this choice we take the $L^\infty$ estimates of all resolvents belonging to dependent momenta except for $k^{(2)}_{l^{'}_2}$, and then integrate out all free $k\upo_l$ with $l$ from $\oln+1$ to $l_2+1$ and all free $k\upt_l$ with $l$ from $\oln+1$ to $l^{'}_2+1$ (note that by construction $k^{(2)}_{l^{'}_2}$ does not depend on any of them), using the usual $L^1$ estimates. 
During this procedure, or lateron, whenever all $k^{(j)}_{l}$ with $l>n_1+1$ are taken care of, we will take the $\beta_j$  ($j=1,2$) integrals with the usual estimates. All remaining free momenta can be integrated over in an arbitrary order, as long as one pays attention to the dependence of $k^{(2)}_{l^{'}_2}$. By construction,
\begin{equation}
k\upo_{l_2-1}=k\upo_{l_1}+\tilde{\xi}\upo
\end{equation}
with 
\begin{equation}
\tilde{\xi}\upo=\begin{cases}\xi\upo &\mbox{if } l_1\leq n_1<l_2-2\\ 0 &\mbox{otherwise.}\end{cases}
\end{equation}
Therefore
\begin{equation}
\begin{split}
k^{(2)}_{l^{'}_2}&=k^{(2)}_{l^{'}_2-1}+k\upo_{l_2}-k\upo_{l_2-1}\\
&=k\upo_{l_2}-k\upo_{l_1}+k^{(2)}_{l^{'}_2-1}-\tilde{\xi}\upo\\
&=k\upo_{l_2}-k\upo_{l_1}+q,
\end{split}
\end{equation}
with $q$ depending only on $\tilde{\xi}\upo$ and free momenta $k\upo_l$ with $l<l_1$ and $k\upt_l$ with $l<l^{'}_2$. We are now ready to estimate the $k\upo_{l_1}$, $k\upo_{l_2}$ and $k^{(2)}_{l^{'}_2}$ resolvents by
\begin{equation}
\label{ct_threeres}
\begin{split}
\sup_{\gamma_i\in I}&\sup_{q\in\tthree}\int_{\left(\tthree\right)^2}\dd k\upo_{l_1} \dd k\upo_{l_2} \frac{1}{\left|e\left(k\upo_{l_1}\right)-\gamma_1-i\eps\right|\left|e\left(k\upo_{l_2}\right)-\gamma_2-i\eps\right|}\\&\hspace{1cm}\times\frac{1}{\left|e\left( k\upo_{l_2}-k\upo_{l_1}+q\right)-\gamma_3+i\eps\right|}\\&\hspace{0.5cm}\leq\eps^{-4/5}\left|\log\eps\right|^3.
\end{split}
\end{equation}
Then we take $L^1$ estimates for all remaining free momenta, and finally treat the integrals over $\alpha_j$, $k^{(j)}_0$, and $\xi^{(j)}$ as before. For the summand containing $\left|\widehat{\phi_0}\left(k\upo_{\oln+1}\right)\right|^2\left|{\widehat{\phi_0}\left(k\upt_{\oln+1}\right)}\right|^2$, the argument works exactly the same, but with all definitions and integrations made "from left to right", now. The cross-over situation, with factors $\left|\widehat{\phi_0}\left(k\upo_0\right)\right|^2\left|{\widehat{\phi_0}\left(k\upt_{\oln+1}\right)}\right|^2$ or $\left|\widehat{\phi_0}\left(k\upo_{\oln+1}\right)\right|^2\left|{\widehat{\phi_0}\left(k\upt_0\right)}\right|^2$, is slightly different; we will concentrate on the summand with $\left|\widehat{\phi_0}\left(k\upo_0\right)\right|^2\left|{\widehat{\phi_0}\left(k\upt_{\oln+1}\right)}\right|^2$. In this case, a momentum $k^{(j)}_l$ is classified as \emph{free} either if $j=1$, $l\notin\lbrace0,n_1+1\rbrace$ and there is an internal delta $\delta\left(k^{(1)}_l-k^{(1)}_{l-1}+k^{(1)}_i-k^{(1)}_{i-1}\right)$, with $l<i$, or if $j=1$ and there is a transfer delta $\delta\left(k^{(1)}_l-k^{(1)}_{l-1}-k^{(2)}_i+k^{(2)}_{i-1}\right)$, or if $j=2$, $l\notin\lbrace n_1,\oln+1\rbrace$ and there is an internal delta $\delta\left(k^{(2)}_l-k^{(2)}_{l+1}+k^{(2)}_i-k^{(2)}_{i+1}\right)$, with $l>i$. All other $k^{(j)}_l$ with $l\notin\lbrace0,n_1+1\rbrace$ ($j=1$) or $l\notin\lbrace n_1,\oln+1\rbrace$ ($j=2$) are \emph{dependent}. In this cross-over case, we utilize the fact that the transfer contractions are not anti-parallel, either, so there exists a pair of transfer deltas $\delta\left(k^{(1)}_{l_1}-k^{(1)}_{l_1-1}-k^{(2)}_{l^{'}_1+1}+k^{(2)}_{l^{'}_1}\right)$ and $\delta\left(k^{(1)}_{l_2}-k^{(1)}_{l_2-1}-k^{(2)}_{l^{'}_2+1}+k^{(2)}_{l^{'}_2}\right)$ such that $l_1<l_2$ and $l^{'}_1<l^{'}_2$. (One can visualize that these two transfer contractions are "crossing" under the given circumstances by bringing the two $\left|\widehat{\phi_0}\right|^2$ factors to the same side of the graph by rotating the second one-particle line by $180^\circ$, compare Fig. \ref{fig_crossing}.) Given the choice of several such pairs, we first choose the maximal $l^{'}_2$, and then, with given $l^{'}_2$ and thus $l_2$, select the largest possible $l_1<l_2$.
For this choice we take the $L^\infty$ estimates of all resolvents belonging to dependent momenta except for $k^{(2)}_{l^{'}_2}$, and then integrate out all free $k\upo_l$ with $l$ from $\oln+1$ to $l_2+1$ and all free $k\upt_l$ with $l$ from $0$ to $l^{'}_2-1$ (note that by construction $k^{(2)}_{l^{'}_2}$ does not depend on any of them), using the usual $L^1$ estimates. 
During this procedure, or lateron, whenever all $k^{(1)}_{l}$ with $l>n_1+1$ and $k\upt_l$ with $l<n_1$ are taken care of, we will take the $\beta_1$ and $\alpha_2$  integrals with the usual estimate, $|\log\eps|$ each. All remaining free momenta can be integrated over in an arbitrary order, as long as one pays attention to the dependence of $k^{(2)}_{l^{'}_2}$. As before,
\begin{equation}
k\upo_{l_2-1}=k\upo_{l_1}+\tilde{\xi}\upo
\end{equation}
with 
\begin{equation}
\tilde{\xi}\upo=\begin{cases}\xi\upo &\mbox{if } l_1\leq n_1<l_2-2\\ 0 &\mbox{otherwise.}\end{cases}
\end{equation}
Therefore
\begin{equation}
\begin{split}
k^{(2)}_{l^{'}_2}&=k^{(2)}_{l^{'}_2+1}-k\upo_{l_2}+k\upo_{l_2-1}\\
&=k\upo_{l_1}-k\upo_{l_2}+k^{(2)}_{l^{'}_2+1}+\tilde{\xi}\upo\\
&=k\upo_{l_1}-k\upo_{l_2}+p,
\end{split}
\end{equation}
with $p$ depending only on $\tilde{\xi}\upo$ and free momenta $k\upo_l$ with $l<l_1$ and $k\upt_l$ with $l>l^{'}_2$. This time we have the integral of the three resolvents for $k\upo_{l_1}$, $k\upo_{l_2}$ and $k^{(2)}_{l^{'}_2}$
\begin{equation}
\begin{split}
\sup_{\gamma_i\in I}&\sup_{p\in\tthree}\int_{\left(\tthree\right)^2}\dd k\upo_{l_1} \dd k\upo_{l_2} \frac{1}{\left|e\left(k\upo_{l_1}\right)-\gamma_1-i\eps\right|\left|e\left(k\upo_{l_2}\right)-\gamma_2-i\eps\right|}\\&\hspace{1cm}\times\frac{1}{\left|e\left( k\upo_{l_1}-k\upo_{l_2}+p\right)-\gamma_3+i\eps\right|}\\&\hspace{0.5cm}\leq\eps^{-4/5}\left|\log\eps\right|^3.
\end{split}
\end{equation}
After $L^1$ estimates for all remaining free momenta one can treat the integrals over $\alpha_1$, $k^{(1)}_0$, and $\xi^{(1)}$ as well as $\beta_2$, $k^{(2)}_{\oln+1}$, and $\xi^{(2)}$ as before. Now we have understood how to treat all four different summands, and for each of them have obtained the same estimates as if we had bounded
\begin{equation}
\int_{(\tthree)^m}\dd u\left|\ampl_1\left(u,\pi_1\right)\right|\sup_{u^{'}}\left|\ampl_2\left(u^{'},\pi_2\right)\right|,
\end{equation}
with the standard estimates from Lemma \ref{lem_factor}, except for three exceptions each, namely the resolvents belonging to $k\upo_{l_1}$, $k\upo_{l_2}$ and $k^{(2)}_{l^{'}_2}$ (two free, one dependent momentum), which did not yield $\eps^{-1}|\log\eps|^2$ but $\eps^{-4/5}\left|\log\eps\right|^3$. So, again, we improve the basic estimate from Corollary \ref{cor_basic} by $\eps^{1/5}|\log\eps|$.
\end{proof}
\begin{figure}[htb] 
\centering 
\def\svgwidth{400pt} 
\input{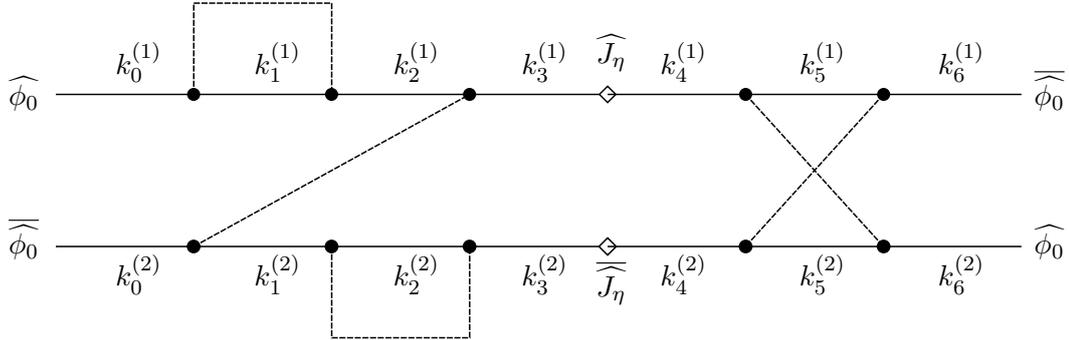} 
\caption{Crossing transfer contractions. The intersection of the two transfer contractions on the right will disappear if we rotate the second one-particle line by $180^\circ$, but then the transfer contraction on the left will produce new crossings.}
\label{fig_crossing} 
\end{figure} 
\label{crossing_transfer}
\subsection{Parallel and antiparallel transfer contractions}
\begin{lma}
\label{lem_para_anti}
Let $\pi\in\bigpi$ be a graph with no generalized crossings on either one-particle line, and with parallel or anti-parallel transfer contractions. Then there is a $c<\infty$ just depending on $J$ such that
\begin{equation}
\label{est_para}
\left|\ampl(\pi)\right|\leq e^{4\eps t}\lambda^{2\oln}\eps^{\frac{1}{5}-\oln}\left|c\log\eps\right|^{\oln+5}\left\Vert\phi_0\right\Vert_{\ell^2}^4
\end{equation}
for all $0<\eps\leq1/3$.
\end{lma}
\begin{proof}
First, let $\pi$ be a graph with parallel transfer contractions, so there are $m$ transfer contractions $\delta\left(k^{(1)}_{l_1+1}-k^{(1)}_{l_1}-k^{(2)}_{l_1^{'}+1}+k^{(2)}_{l_1^{'}}\right)$ to $\delta\left(k^{(1)}_{l_m+1}-k^{(1)}_{l_m}-k^{(2)}_{l_m^{'}+1}+k^{(2)}_{l_m^{'}}\right)$ with $l_1<...<l_m$ and $l_1^{'}<...<l_m^{'}$. 
\begin{figure}[htb] 
\centering 
\def\svgwidth{400pt} 
\input{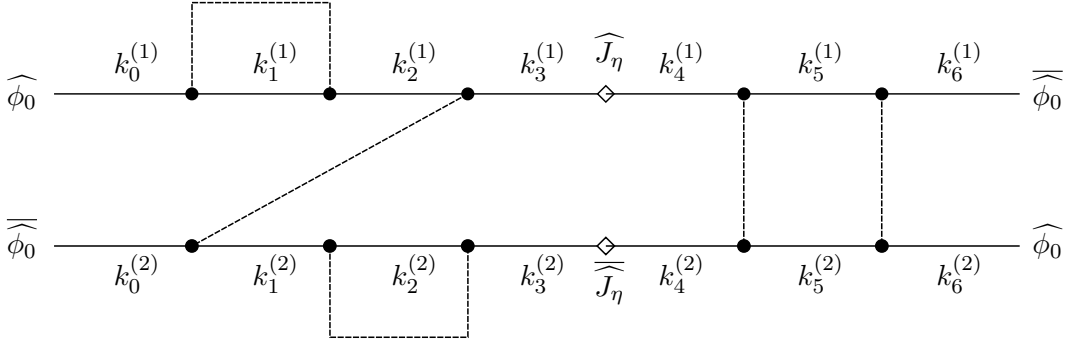} 
\caption{Three parallel transfer contractions, with no intersection at all.}
\label{fig_parallel} 
\end{figure} 
We arbitrarily single out the leftmost transfer delta and define the set
\begin{equation}
\label{defcab}
\caB=\left\lbrace\left(k\upo,k\upt\right)\in\left(\tthree\right)^{2\oln+4}:\dist\left(k^{(1)}_{l_1}-k^{(2)}_{l_1^{'}},\caE\right)<\eps^{2/5}\right\rbrace,
\end{equation}
with $\caE=\left((0,0,0),\left(\frac{1}{2},\frac{1}{2},\frac{1}{2}\right)\right)$, and split up 
\begin{equation}
\left|\ampl(\pi)\right|\leq\ampl_{\caB}(\pi)+\ampl_{\caB^c}(\pi),
\end{equation}
where we denoted
\begin{equation}
\label{amplb}
\begin{split}
\ampl_{\caB}(\pi)=&\frac{e^{4\eps t}\lambda^{2\oln}}{(2\pi)^4}\int_{I}\left|\dd\alpha_1\right|\int_{\overline{I}}\left|\dd\beta_1\right|\int_{\overline{I}}\left|\dd\alpha_2\right|\int_{I}\left|\dd\beta_2\right|\\
&\int_{\caB}\dd k\upo\dd k\upt\delta_\pi\left(k\upo{},k\upt\right)\\
&\int_{\bbR^{3}}\dd\xi^{(1)}\int_{\bbR^{3}}\dd\xi^{(2)}\delta\left(k\upo_{n_1+1}-k\upo_{n_1}-\xi\upo\right)\delta\left(k\upt_{n_1+1}-k\upt_{n_1}-\xi\upt\right)\\
&\hspace{1cm}\left|{\widehat{J_\eta}\left(\xi\upo,\left(k_{n_1}\upo+k_{n_1+1}\upo\right)/2\right)}\overline{\widehat{J_\eta}\left(\xi\upt,\left(k_{n_1}\upt+k_{n_1+1}\upt\right)/2\right)}\right|\\
&\hspace{1cm}\left|K^{(1)}\left(k^{(1)},\alpha_1,\beta_1,\eps\right)K^{(2)}\left(k^{(2)},\alpha_2,\beta_2,\eps\right)\right|\\
&\hspace{1cm}\left|\widehat{\phi_0}\left(k\upo_0\right)\right|\left|{\widehat{\phi_0}\left(k\upo_{\oln+1}\right)}\right|\left|{\widehat{\phi_0}\left(k\upt_0\right)}\right|\left|\widehat{\phi_0}\left(k\upt_{\oln+1}\right)\right|,\\
\end{split}
\end{equation}
and similarly $\ampl_{\caB^c}(\pi)$ with the integral taken over $\caB^c$ instead of $\caB$. We first study $\ampl_{\caB}(\pi)$. It is important that in this case, we will \emph{not} estimate  $\left|\widehat{\phi_0}\left(k\upo_0\right)\right|\left|{\widehat{\phi_0}\left(k\upo_{\oln+1}\right)}\right|$ and similar products by sums of squares. Notice that if we take
\begin{equation}
\tilde{\xi}\upo=\begin{cases}\xi\upo &\mbox{if } l_1>n_1\\ 0 &\mbox{otherwise}\end{cases}
\end{equation}
and
\begin{equation}
\tilde{\xi}\upt=\begin{cases}\xi\upt &\mbox{if } l^{'}_1>n_1\\ 0 &\mbox{otherwise,}\end{cases}
\end{equation}
then
\begin{equation}
\caB=\left\lbrace\left(k\upo,k\upt\right)\in\left(\tthree\right)^{2\oln+4}:\dist\left(k^{(1)}_{0}+\tilde{\xi}\upo-k^{(2)}_{0}-\tilde{\xi}\upt,\caE\right)<\eps^{2/5}\right\rbrace.
\end{equation}
Later, we will therefore use $k^{(1)}_{0}$ and $w$ as integration variables, with
\begin{equation}
\label{constraintw}
\dist(w,\caE)<\eps^{2/5}
\end{equation}
and have 
\begin{equation}
k^{(2)}_{0}=k^{(1)}_{0}+w+\tilde{\xi}\upo-\tilde{\xi}\upt.
\end{equation}
We also slightly modify the classification of $k^{(j)}_l$ (with $l\notin\lbrace0,n_1+1\rbrace$ into free and dependent momenta. We still call $k^{(j)}_l$ \emph{free} if there is an internal delta\\ $\delta\left(k^{(j)}_l-k^{(j)}_{l-1}+k^{(j)}_i-k^{(j)}_{i-1}\right)$, with $j=1,2$ and $l<i$, or if it is one of the $k^{(1)}_{l_j+1}$, $j=1,...,m-1$, but all other momenta, now including $k^{(1)}_{l_m+1}$, are considered \emph{dependent}. Thus, we have one more dependent momentum, and consequently freed up one new integration variable, which we take to be
\begin{equation}
u_m=k^{(1)}_{l_m+1}-k^{(1)}_{l_m}=k^{(2)}_{l_m^{'}+1}-k^{(2)}_{l_m^{'}},
\end{equation}
and we can rewrite the last line of (\ref{amplb}) as
\begin{equation}
\left|\widehat{\phi_0}\left(k\upo_0\right)\right|\left|{\widehat{\phi_0}\left(k^{(1)}_{l_m}+u_m\right)}\right|\left|{\widehat{\phi_0}\left(k^{(1)}_{0}+w+\tilde{\xi}\upo-\tilde{\xi}\upt\right)}\right|\left|\widehat{\phi_0}\left(k^{(2)}_{l_m^{'}}+u_m\right)\right|.
\end{equation}
After taking $L^\infty$ estimates of all dependent resolvents, and $L^1$ bounds for all free resolvents $k^{(j)}_l$, with $j=1$ and $l>l_m$, or $j=2$ and $l>l^{'}_m$, the only remaining integration variable right of the $m$-th transfer contraction is $u_m$; integrating over $u_m$ produces a factor $\left\Vert\phi_0\right\Vert_{\ell^2}^2$. (During this process, or later, depending on when we reach $n_1+1$ we will also take the $\beta_j$ ($j=1,2$) integrals with the usual $|\log\eps|$ estimates.) Then we integrate out all remaining free resolvents with $L^1$ estimates, the $\alpha_j$ ($j=1,2$) integrals with a bound $|\log\eps|$, each. Now take the $k\upo_0$ integral for another $\left\Vert\phi_0\right\Vert_{\ell^2}^2$ factor. The $w$ integral is completely "flat", as all wave functions and resolvents have been taken care of, so by (\ref{constraintw}), it will yield a factor $\eps^{6/5}$. Then the $\xi\upo$ and $\xi\upt$ integral will contribute a $J$-dependent constant. All together, for the estimate of $\ampl_{\caB}(\pi)$ we have one $L^\infty$ factor more, one $L^1$ factor less (as there is one dependent resolvent more), and a geometrical factor $\eps^{6/5}$, so we improve the standard estimate by $\eps^{1/5}|\log\eps|^{-1}$,
\begin{equation}
\ampl_{\caB}(\pi)\leq e^{4\eps t}\lambda^{2\oln}\eps^{\frac{1}{5}-\oln}\left|c\log\eps\right|^{\oln+3}\left\Vert\phi_0\right\Vert_{\ell^2}^4.
\end{equation}
For the estimate of $\ampl_{\caB^c}(\pi)$, we now bound
\begin{equation}
\label{prodtosquare}
\begin{split}
&\left|\widehat{\phi_0}\left(k\upo_0\right)\right|\left|{\widehat{\phi_0}\left(k\upo_{\oln+1}\right)}\right|\left|{\widehat{\phi_0}\left(k\upt_0\right)}\right|\left|\widehat{\phi_0}\left(k\upt_{\oln+1}\right)\right|\\
&\hspace{5mm}\leq\left(\frac{1}{2}\left|\widehat{\phi_0}\left(k\upo_0\right)\right|^2+\frac{1}{2}\left|\widehat{\phi_0}\left(k\upo_{\oln+1}\right)\right|^2\right)\left(\frac{1}{2}\left|{\widehat{\phi_0}\left(k\upt_0\right)}\right|^2+\frac{1}{2}\left|\widehat{\phi_0}\left(k\upt_{\oln+1}\right)\right|^2\right)
\end{split}
\end{equation}
again. We first consider the summand containing $\left|\widehat{\phi_0}\left(k\upo_0\right)\right|^2\left|{\widehat{\phi_0}\left(k\upt_0\right)}\right|^2$. Our free integration variables will be almost as above, a momentum $k^{(j)}_l$ is free if there is an internal delta $\delta\left(k^{(j)}_l-k^{(j)}_{l-1}+k^{(j)}_i-k^{(j)}_{i-1}\right)$, with $j=1,2$ and $l<i$, or if it is one of the $k^{(1)}_{l_j+1}$, $j=2,...,m$, but we will write $k^{(1)}_{l_1+1}=k^{(1)}_{l_1}+u_1$ and $k^{(2)}_{l^{'}_1+1}=k^{(2)}_{l^{'}_1}+u_1$ as functions of a new integration variable $u_1$, and finally all other $k^{(j)}_l$ with $l\notin\lbrace0,n_1+1\rbrace$ are dependent. We take $L^\infty$ estimates of all dependent momenta resolvents and integrate out all free $k^{(1)}_l$ with $l>l_1+1$ and $k^{(2)}_l$ with $l>l^{'}_1+1$. Next, for the integral over the $k^{(1)}_{l_1+1}$ and $k^{(2)}_{l^{'}_1+1}$ resolvents
\begin{equation}
\begin{split}
\sup_{\alpha,\beta\in I}&\int_{\tthree} \frac{\dd u_1}{\left|e\left(k^{(1)}_{l_1}+u_1\right)-\alpha-i\eps\right|\left|e\left(k^{(2)}_{l^{'}_1}+u_1\right)-\beta-i\eps\right|}\\&\leq C\left( \eps^{-4/5}+\frac{\eps^{-2/5}}{\dist(k^{(1)}_{l_1}-k^{(2)}_{l^{'}_1},\caE)}\right)\left|\log\eps\right|^2\\
&\leq C\eps^{-4/5}|\log\eps|^2
\end{split}
\end{equation}
by Lemma \ref{lemma_2res} and the definition (\ref{defcab}) of $\caB^c$. Then we integrate out all remaining free momenta, having taken the $\beta_j$ integrals at the appropriate time. Then, the $\alpha_j$, $k^{(j)}_0$ and $\xi^{(j)}$ integrals will be estimated as usual. Alltogether, in contrast to the standard estimate, we have one $L^\infty$ and one $L^1$ estimate less, and a factor $\eps^{-4/5}|\log\eps|^2$ instead, thus improving Corollary \ref{cor_basic} by $\eps^{1/5}|\log\eps|$. The case with $\left|\widehat{\phi_0}\left(k\upo_{\oln+1}\right)\right|^2\left|{\widehat{\phi_{0}}\left(k\upt_{\oln+1}\right)}\right|^2$ can be treated completely analogously by using Lemma \ref{lemma_2res}. For the cross-over cases, i.e. a factor $\left|\widehat{\phi_0}\left(k\upo_{0}\right)\right|^2\left|{\widehat{\phi_{0}}\left(k\upt_{\oln+1}\right)}\right|^2$ or $\left|\widehat{\phi_0}\left(k\upo_{\oln+1}\right)\right|^2\left|{\widehat{\phi_{0}}\left(k\upt_0\right)}\right|^2$, we assume for a moment $m>1$ (several transfer contractions). Then just bounding $1_{\caB^c}\leq1$ immediately produces an improvement $\eps^{1/5}|\log\eps|$ by the methods of section \ref{crossing_transfer}. (Several parallel transfer lines are ``seen" as crossing by ``cross-over" $\left|\widehat{\phi_0}\right|^2$).
\begin{figure}[htb] 
\centering 
\def\svgwidth{400pt} 
\input{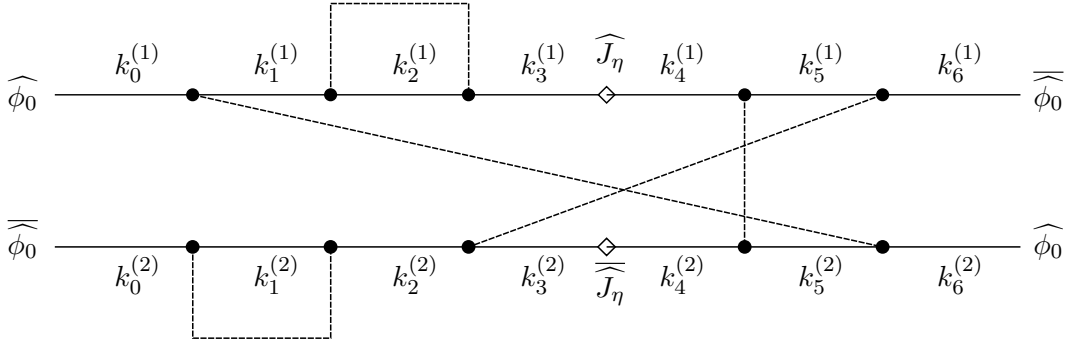} 
\caption{Three anti-parallel transfer contractions. Their intersection will completely disappear by rotating the second one-particle line by $180^\circ$.}
\label{fig_anti} 
\end{figure}

Before we study the special case $m=1$, note that almost the same methods apply to anti-parallel transfer lines. If $\pi\in\bigpi$ with anti-parallel transfer lines, there are $m$ transfer contractions $\delta\left(k^{(1)}_{l_1+1}-k^{(1)}_{l_1}-k^{(2)}_{l_1^{'}+1}+k^{(2)}_{l_1^{'}}\right)$ to $\delta\left(k^{(1)}_{l_m+1}-k^{(1)}_{l_m}-k^{(2)}_{l_m^{'}+1}+k^{(2)}_{l_m^{'}}\right)$ with $l_1<...<l_m$ but $l_1^{'}>...>l_m^{'}$. We just repeat the above reasoning "with the second one-particle line rotated" and define
\begin{equation}
\label{defcac}
\caC=\left\lbrace\left(k\upo,k\upt\right)\in\left(\tthree\right)^{2\oln+4}:\dist\left(k^{(1)}_{l_1}+k^{(2)}_{l_1^{'}+1},\caE\right)<\eps^{2/5}\right\rbrace.
\end{equation}
The bound
\begin{equation}
\ampl_{\caC}(\pi)\leq e^{4\eps t}\lambda^{2\oln}\eps^{\frac{1}{5}-\oln}\left|c\log\eps\right|^{\oln+3}\left\Vert\phi_0\right\Vert_{\ell^2}^4.
\end{equation}
is obtained for all $m\geq1$ as before, while for $\ampl_{\caC^c}$, after taking the estimate (\ref{prodtosquare}) , the summands with the factors $\left|\widehat{\phi_0}\left(k\upo_{0}\right)\right|^2\left|{\widehat{\phi_{0}}\left(k\upt_{0}\right)}\right|^2$ or $\left|\widehat{\phi_0}\left(k\upo_{\oln+1}\right)\right|^2\left|{\widehat{\phi_{0}}\left(k\upt_{\oln+1}\right)}\right|^2$ can be bounded by  
\begin{equation}
 e^{4\eps t}\lambda^{2\oln}\eps^{\frac{1}{5}-\oln}\left|c\log\eps\right|^{\oln+5}\left\Vert\phi_0\right\Vert_{\ell^2}^4.
\end{equation}
by the methods of section \ref{crossing_transfer}, if we assume $m>1$. (Several anti-parallel transfer contractions are seen as crossing by ``cis" $\left|\widehat{\phi_0}\right|^2$). On the other hand, for all $m\geq1$, the cross-over cases $\left|\widehat{\phi_0}\left(k\upo_{0}\right)\right|^2\left|{\widehat{\phi_{0}}\left(k\upt_{\oln+1}\right)}\right|^2$ or $\left|\widehat{\phi_0}\left(k\upo_{\oln+1}\right)\right|^2\left|{\widehat{\phi_{0}}\left(k\upt_0\right)}\right|^2$ are now estimated using Lemma \ref{lemma_2res} to improve the standard estimate by $\eps^{1/5}|\log\eps|$.

Now the only cases not completely understood yet are $\pi\in\bigpi$ with $m=1$, thus only one transfer contraction, which qualifies both as parallel and antiparallel. With $\caB$ and $\caC$ defined as above, one easily derives
\begin{equation}
\begin{split}
\left|\ampl(\pi)\right|&\leq\ampl_{\caB}(\pi)+\ampl_{\caC}(\pi)+\ampl_{\caB^{c}\cap\caC^{c}}(\pi)\\
&\leq e^{4\eps t}\lambda^{2\oln}\eps^{\frac{1}{5}-\oln}\left|c\log\eps\right|^{\oln+5}\left\Vert\phi_0\right\Vert_{\ell^2}^4
\end{split}
\end{equation}
by the above methods. In this case, all four (two cross-over, two cis) contributions to $\ampl_{\caB^{c}\cap\caC^{c}}(\pi)$ are estimated by the two-resolvent integral from Lemma \ref{lemma_2res}.

Collecting all estimates and possibly redefining $c$, we arrive at (\ref{est_para}) for all parallel or anti-parallel $\pi\in\bigpi$.
\end{proof}

\label{para_anti}
\section{Proof of the main theorem}
For the time being, consider initial states $\phi^{(\eta)}_0$ which in addition to the assumption (\ref{boundedprob}) satisfy $\widehat{\phi^{(\eta)}_0}\in C^{\infty}\left(\bbT^3\right)$ for all $\eta>0$. In that case, all estimates derived in sections \ref{graph_exp} and \ref{graph_estimates} are applicable if we choose, for $t=T/\eta$, 
\begin{equation}
\label{fix_eps}
\begin{split}
\eps&=\frac{1}{3+t}\\
N&=\left\lfloor\frac{a|\log\eps|}{|\log|\log\eps||}\right\rfloor\\
\kappa&=\left\lceil |\log\eps|^b\right\rceil,
\end{split}
\end{equation}
with $a,b>0$. Furthermore, we assume $\lambda\leq\frac{1}{2}$.

With these settings, for arbitrarily small $\delta$, there is a constant $c$ only depending on $a,b,\delta$, such that
\begin{equation}
\begin{split}
c^{-1}\eps^{-a+\delta}&\leq N!\leq N^N\leq c\eps^{-a}\\
|\log\eps|^N&\leq\eps^{-a}\\
\kappa^{-N}&\leq c\eps^{ba-\delta}.
\end{split}
\end{equation}
Thus, with (\ref{est_remainder}), there is a $c$ depending only on $a,b,T$ and an arbitrarily small $\delta>0$, such that
\begin{equation}
\bbE\left[\left\Vert R_{N,t}\right\Vert_{\ell^2}^2\right]\leq c\eps^{-\delta}\left(\eps^{a/2}+\eps^{1/5-8a}+\eps^{-2+ba-8a}+\eps^{1-24a}\right).
\end{equation}
On the other hand, we note that
\begin{equation}
\sharp\Pi^{conn}_{n_1,n_2}\leq2^{\overline{n}}\overline{n}!.
\end{equation}
with $\oln=n_1+n_2$ and can plug the estimates from Lemmas \ref{lem_gen_crossing}, \ref{lem_crossing_transfer}, \ref{lem_para_anti} into (\ref{sum_covari}) to see that
\begin{equation}
\begin{split}
&\hspace{-1cm}\vari{\bra J,W^{\eta}\left[\phi\upm_{t}\right]\ket}\\&\leq (N+1)^2C_{J,T}^N\sum_{n_1,n_2=0}^N\left(\sharp\Pi^{conn}_{n_1,n_2}\right)\eps^{1/5-2a}\\
&\leq C(J,T,a,\delta)\eps^{-\delta}\eps^{1/5-4a}
\end{split}
\end{equation}
with arbitrarily small $\delta>0$.

We now first want to verify Theorem \ref{theorem_main} for $r=1$, and apply Lemma \ref{est_bilin} to estimate
\begin{equation}
\begin{split}
\bbE&\left[\left|\bra J,W^{\eta}\left[\phi^{(\eta)}_{T/\eta}\right]\ket-\bbE\bra J,W^{\eta}\left[\phi^{(\eta)}_{T/\eta}\right]\ket\right|\right]\\
&\leq\bbE\left[\left|\bra J,W^{\eta}\left[R_{N,t}\right]\ket-\bbE\bra J,W^{\eta}\left[R_{N,t}\right]\ket\right|\right]\\
&\hspace{1cm}+\bbE\left[\left|\bra J,W^{\eta}\left[R_{N,t},\phi\upm_t\right]\ket-\bbE\bra J,W^{\eta}\left[R_{N,t},\phi\upm_t\right]\ket\right|\right]\\
&\hspace{1cm}+\bbE\left[\left|\bra J,W^{\eta}\left[\phi\upm_t,R_{N,t}\right]\ket-\bbE\bra J,W^{\eta}\left[\phi\upm_t,R_{N,t}\right]\ket\right|\right]\\
&\hspace{1cm}+\bbE\left[\left|\bra J,W^{\eta}\left[\phi\upm_t\right]\ket-\bbE\bra J,W^{\eta}\left[\phi\upm_t\right]\ket\right|\right]\\
&\leq C_J\left(2\bbE\left[\Vert R_{N,t}\Vert_{\ell^2}^2\right]+4\left(\bbE\left[\Vert \phi\upm_t\Vert_{\ell^2}^2\right]\bbE\left[\Vert R_{N,t}\Vert_{\ell^2}^2\right]\right)^{1/2}\right),\\
&\hspace{1cm}+\vari{\bra J,W^{\eta}\left[\phi\upm_{t}\right]\ket}^{1/2},
\end{split}
\end{equation}
with $C_J<\infty$ just depending on our choice of the test function $J$. Noticing that
\begin{equation}
\Vert \phi\upm_t\Vert_{\ell^2}^2\leq\left(1+\Vert R_{N,t}\Vert_{\ell^2}\right)^2
\end{equation}
and choosing $a=\frac{2}{85}$, $b=100$ and $\delta$ small enough, we find that there is a $\tilde{C}(J,T)<\infty$ such that
\begin{equation}
\bbE\left[\left|\bra J,W^{\eta}\left[\phi^{(\eta)}_{T/\eta}\right]\ket-\bbE\bra J,W^{\eta}\left[\phi^{(\eta)}_{T/\eta}\right]\ket\right|\right]\leq \tilde{C}(J,T)\lambda^{\frac{1}{90}}.
\end{equation}
For general $r\geq1$, we apply Lemma \ref{est_bilin} and the unitarity of $\mathrm{e}^{-iH_\omega t}$ to see
\begin{equation}
\bbE\left[\left|\bra J,W^{\eta}\left[\phi^{(\eta)}_{T/\eta}\right]\ket-\bbE\bra J,W^{\eta}\left[\phi^{(\eta)}_{T/\eta}\right]\ket\right|^r\right]\leq \left(2C_J\right)^{r-1}\tilde{C}(J,T)\lambda^{\frac{1}{90}}.
\end{equation}
Thus we have shown (\ref{fluct_tozero}) for $c=\frac{1}{90}$, initial states with smooth $\widehat{\phi}^{(\eta)}_0$ and
\begin{equation}
C(J,T)=\max\left\lbrace\tilde{C}(J,T),2C_J\right\rbrace.
\end{equation}

Now let a sequence of general initial states $\phi^{(\eta)}_0$ fulfilling (\ref{boundedprob}) be given and approximate them by $\tilde{\phi}^{(\eta)}_0$ with $\widehat{\tilde{\phi}^{(\eta)}_0}\in C^{\infty}\left(\bbT^3\right)$ such that 
\begin{equation}
\left\Vert\tilde{\phi}^{(\eta)}_0-{\phi}^{(\eta)}_0\right\Vert^2_{\ell^2\left(\bbZ^3\right)}\leq\eta,
\end{equation}
with (\ref{boundedprob}) holding for $\tilde{\phi}^{(\eta)}_0$ as well. For the smooth initial states we have
\begin{equation}
\left(\bbE\left[\left|\bra J,W^{\eta}\left[\tilde{\phi}^{(\eta)}_{T/\eta}\right]\ket-\bbE\bra J,W^{\eta}\left[\tilde{\phi}^{(\eta)}_{T/\eta}\right]\ket\right|^r\right]\right)^{1/r}\leq C(J,T) \lambda^{\frac{1}{cr}},
\end{equation}
while by unitarity of the time evolution and Lemma \ref{est_bilin},
\begin{equation}
\left|\bra J,W^{\eta}\left[\tilde{\phi}^{(\eta)}_{T/\eta}\right]\ket-\bra J,W^{\eta}\left[{\phi}^{(\eta)}_{T/\eta}\right]\ket\right|\leq 2 C_J\sqrt{\eta}=2C_J\lambda
\end{equation}
almost surely. This proves (\ref{fluct_tozero}), for general initial states, possibly after redefining the constant $C(J,T)$, but still with $c=1/90$.

\label{proof_main}
\section{Almost sure convergence}
\label{almostsure}
To prove Theorem \ref{thm_almostsure}, fix any $T>0$ and $J\in\caS\left(\bbR^3\times\bbT^3\right)$, and observe
\begin{equation}
\begin{split}
&\sup_{\tau\in[0,T]}\left|\bra J,W^\eta\left[\mathrm{e}^{-iH_\omega\tau/\eta}\phi_0^{(\eta)}\right]\ket-\bra J,\mu_\tau\ket\right|\\
&\hspace{1cm}\leq\sup_{\tau\in[0,T]}\left|\bra J,W^\eta\left[\mathrm{e}^{-iH_\omega\tau/\eta}\phi_0^{(\eta)}\right]\ket-\bbE\bra J,W^\eta\left[\mathrm{e}^{-iH_\omega\tau/\eta}\phi_0^{(\eta)}\right]\ket\right|\\
&\hspace{2cm}+\sup_{\tau\in[0,T]}\left|\bbE\bra J,W^\eta\left[\mathrm{e}^{-iH_\omega\tau/\eta}\phi_0^{(\eta)}\right]\ket-\bra J,\mu_\tau\ket\right|.\\
\end{split}
\end{equation}
The last line is deterministic and vanishes as $\eta$ goes to zero by Theorem \ref{theorem_convergence_exp} (for the $\sup_\tau$, note that all error estimates in the proof of Theorem \ref{theorem_convergence_exp} are monotonous in $\tau$, \cite{chen_rs}). To understand the second last line, discretize $\eta\in(0,1]$ by $\eta_{n}=n^{-\beta}$, $n\in\bbN$ for some $\beta>0$ yet to be determined. For any given $h\in\caS\left(\bbR^3\right)$ and real $S\in\caS\left(\bbR^3\right)$, we have for initial states $\phi_0^{(\eta)}$ of the form (\ref{chen_init})
\begin{equation}
\left\Vert\frac{\dd}{\dd \eta}\phi_0^{(\eta)}\right\Vert_{\ell^2}\leq\frac{C_{h,S}}{\eta^2},
\end{equation}
and thus
\begin{equation}
\label{phi_cts}
\sup_{\eta\in\left[\eta_{n+1},\eta_{n}\right]}\left\Vert\phi_0^{(\eta)}-\phi_0^{(\eta_{n})}\right\Vert_{\ell^2}\leq C_{h,S}\left(\frac{1}{\eta_{n+1}}-\frac{1}{\eta_{n}}\right).
\end{equation}
Next, for any $\psi\in\ell^2\left(\bbZ^3\right)$, and any $J\in\caS\left(\bbR^3\times\bbT^3\right)$, we note for
\begin{equation}
\bra J,W^{\eta}\left[\psi\right]\ket=\eta^{-3}\int_{\bbR^3\times\tthree}\dd\xi\dd v\overline{\widehat{J}(\xi/\eta,v)}\overline{\widehat{\psi}(v-\xi/2)}\widehat{\psi}(v+\xi/2)
\end{equation}
that
\begin{equation}
\left|\frac{\dd}{\dd \eta}\bra J,W^{\eta}\left[\psi\right]\ket\right|\leq\frac{C_J}{\eta}\left\Vert\psi\right\Vert_{\ell^2}^2.
\end{equation}
Altogether we have
\begin{equation}
\begin{split}
&\sup_{\eta\in\left[\eta_{n+1},\eta_{n}\right]}\sup_{\tau\in[0,T]}\left|\bra J,W^\eta\left[\mathrm{e}^{-iH_\omega\tau/\eta}\phi_0^{(\eta)}\right]\ket-\bbE\bra J,W^\eta\left[\mathrm{e}^{-iH_\omega\tau/\eta}\phi_0^{(\eta)}\right]\ket\right|\\
&\hspace{0.5cm}\leq\sup_{\eta\in\left[\eta_{n+1},\eta_{n}\right]}\sup_{\tau\in[0,T]}\left|\bra J,W^{\eta_n}\left[\mathrm{e}^{-iH_\omega\tau/\eta}\phi_0^{(\eta_n)}\right]\ket-\bbE\bra J,W^{\eta_n}\left[\mathrm{e}^{-iH_\omega\tau/\eta}\phi_0^{(\eta_n)}\right]\ket\right|\\
&\hspace{1cm}+C_{h,S}C_J\left(\left(\frac{1}{\eta_{n+1}}-\frac{1}{\eta_{n}}\right)+\log\eta_n-\log\eta_{n+1}\right),
\end{split}
\end{equation}
with constants only depending on $h$, $S$ and $J$.  The last line deterministically vanishes as $n\rightarrow\infty$ whenever $\beta\in(0,1)$. In the second last line, the continuous parameter $\eta$ appears still in two places, in the scaling of the time $\tau/\eta$, and implicitely in the definition of the Hamiltonian
\begin{equation}
\label{hamiltonianl}
H_\omega=H^{\lambda}_\omega=-\frac{1}{2}\Delta+\lambda V_\omega.
\end{equation}
with $\lambda=\sqrt{\eta}$. Define macroscopic times $\tau_0^{(n)}=0<...<\tau_{m_n}^{(n)}=T$ with a spacing $\tau_{j+1}^{(n)}-\tau_j^{(n)}=\delta_n\leq\eta_nT$ and $m_n=\left\lceil n^\beta\right\rceil$. For $\eta\in\left[\eta_{n+1},\eta_{n}\right]$ and $\tau\in\left[\tau_j^{(n)},\tau_{j+1}^{(n)}\right]$, one has
\begin{equation}
\frac{\tau^{(n)}_{j}}{\eta_n}\leq\frac{\tau}{\eta}\leq\frac{\tau^{(n)}_{j+1}}{\eta_{n+1}}\leq\frac{\tau^{(n)}_{j}}{\eta_n}+T\left(\frac{1}{\eta_{n+1}}-\frac{1}{\eta_{n}}\right) +\frac{\delta_n}{\eta_{n+1}}=\frac{\tau^{(n)}_{j}}{\eta_n}+\nu_n.
\end{equation}

\noindent{}First, only consider
\begin{equation}
\left|\bra J,W^{\eta_n}\left[\mathrm{e}^{-iH^\lambda_\omega \tau^{(n)}_{j}/{\eta_n}}\phi_0^{(\eta_n)}\right]\ket-\bbE\bra J,W^{\eta_n}\left[\mathrm{e}^{-iH^\lambda_\omega \tau^{(n)}_{j}/{\eta_n}}\phi_0^{(\eta_n)}\right]\ket\right|
\end{equation}
with $\lambda\in\left[0,\sqrt{\eta_n}\right]$. In keeping with equations (\ref{defphin}-\ref{defphimain}), denote
\begin{equation}
D_{\omega,l}^\lambda\left(\frac{\tau^{(n)}_{j}}{\eta_n}\right)\phi_0^{(\eta_n)}=(-i\lambda)^l\int_{\bbR_+^{l+1}}\dd s_0...\dd s_l\delta\left(\sum_{j=0}^ns_j-\frac{\tau^{(n)}_{j}}{\eta_n}\right)\prod_{j=0}^{l-1}\left(\mathrm{e}^{-is_jH_0}V_\omega\right)\mathrm{e}^{-is_lH_0}\phi^{(\eta_n)}_0
\end{equation}
for $l=0,...,N$ and the remainder term
\begin{equation}
R_{\omega,N}^\lambda\left(\frac{\tau^{(n)}_{j}}{\eta_n}\right)\phi_0^{(\eta_n)}=\mathrm{e}^{-iH^\lambda_\omega \tau^{(n)}_{j}/{\eta_n}}\phi_0^{(\eta_n)}-\sum_{l=0}^{N}D_{\omega,l}^\lambda\left(\frac{\tau^{(n)}_{j}}{\eta_n}\right)\phi_0^{(\eta_n)}.
\end{equation}

Focussing on the $D_{\omega,l}^\lambda$ terms, we note that
\begin{equation}
\label{lambdahom}
D_{\omega,l}^\lambda\left(\frac{\tau^{(n)}_{j}}{\eta_n}\right)\phi_0^{(\eta_n)}=\left(\frac{\lambda}{\sqrt{\eta_n}}\right)^lD_{\omega,l}^{\sqrt{\eta_n}}\left(\frac{\tau^{(n)}_{j}}{\eta_n}\right)\phi_0^{(\eta_n)}
\end{equation}
holds determinstically, so it suffices to set the disorder parameter to the maximal value $\lambda=\sqrt{\eta_n}$. 

If we choose $N$ as in (\ref{fix_eps}) (with $\eps=\eta_n$), then there is a $c>0$ (one may take $c=1/6$) and $C_{J,T}<\infty$ such that
\begin{equation}
\begin{split}
\label{estsquare}
&\bbE\left[\left|\bra J,W^{\eta_n}\left[D_{\omega,l_1}^{\sqrt{\eta_n}}\left(\frac{\tau^{(n)}_{j}}{\eta_n}\right)\phi_0^{(\eta_n)},D_{\omega,l_2}^{\sqrt{\eta_n}}\left(\frac{\tau^{(n)}_{j}}{\eta_n}\right)\phi_0^{(\eta_n)}\right]\ket\right.^{\vphantom{2}}\right.\\&\hspace{1cm}\left.\left.-\bbE\bra J,W^{\eta_n}\left[D_{\omega,l_1}^{\sqrt{\eta_n}}\left(\frac{\tau^{(n)}_{j}}{\eta_n}\right)\phi_0^{(\eta_n)},D_{\omega,l_2}^{\sqrt{\eta_n}}\left(\frac{\tau^{(n)}_{j}}{\eta_n}\right)\phi_0^{(\eta_n)}\right]\ket\right|^2\right]\leq C_{J,T}\eta_n^{c}
\end{split}
\end{equation}
for all $n\in\bbN$, $0\leq l_1,l_2\leq N$ and $0\leq j\leq m_n$. At the same time,
the random variable
\begin{equation}
\begin{split}
X(\omega)&=\bra J,W^{\eta_n}\left[D_{\omega,l_1}^{\sqrt{\eta_n}}\left(\frac{\tau^{(n)}_{j}}{\eta_n}\right)\phi_0^{(\eta_n)},D_{\omega,l_2}^{\sqrt{\eta_n}}\left(\frac{\tau^{(n)}_{j}}{\eta_n}\right)\phi_0^{(\eta_n)}\right]\ket\\&\hspace{1cm}-\bbE\bra J,W^{\eta_n}\left[D_{\omega,l_1}^{\sqrt{\eta_n}}\left(\frac{\tau^{(n)}_{j}}{\eta_n}\right)\phi_0^{(\eta_n)},D_{\omega,l_2}^{\sqrt{\eta_n}}\left(\frac{\tau^{(n)}_{j}}{\eta_n}\right)\phi_0^{(\eta_n)}\right]\ket
\end{split}
\end{equation}
is contained in $\overline{\caP_{l_1+l_2}}$, as defined in Definition 2.1 of \cite{janson}, that is, $X$ lies in the $L^2(\Omega,\caF,\bbP)$ closure of $\caP_{l_1+l_2}$, with $(\Omega,\caF,\bbP)$ the probability space over which the random potential $V_\omega$ is realized, and $\caP_k$, $k\in\bbN_0$ given as the space of polynomials in single-site values of $V_\omega$ of degree smaller or equal $k$:
\begin{equation}
\caP_k=\mathrm{span}\left\lbrace \prod_{j=1}^{p}V(x_j):x_1,...,x_p\in\bbZ^3, 0\leq p\leq k\right\rbrace.
\end{equation}
To see this, first define a variable $X^L\in\caP_{l_1+l_2}$ by employing a cut-off version $V^L_\omega$ of the random potential as in (\ref{vcutoff}), and then note that $X^L\rightarrow X$ in $L^2(\Omega,\caF,\bbP)$ as $L\rightarrow\infty$. 

By Theorem 5.10 in \cite{janson}, one therefore has 
\begin{equation}
\bbE\left[|X|^q\right]^{1/q}\leq(q-1)^{(l_1+l_2)/2}\bbE\left[|X|^2\right]^{1/2},
\end{equation}
for all $q\geq2$, and by (\ref{estsquare}) and a Markov estimate
\begin{equation}
\label{largdevx}
\bbP\left(X\geq\eta_n^\gamma\right)\leq{C}_{J,T}^q(q-1)^{(l_1+l_2)q/2}\eta_n^{(c/2-\gamma)q}
\end{equation}
for any $0<\gamma<c/2$.

Thus, for the main term, by (\ref{lambdahom}) and (\ref{largdevx}),
\begin{equation}
\label{mainbound}
\begin{split}
&\sup_{j\in\lbrace0,...,m_n\rbrace}\sup_{\lambda\in\left[0,\sqrt{\eta_n}\right]}\left|\bra J,W^{\eta_n}\left[\sum_{l=0}^ND_{\omega,l}^{\lambda}\left(\frac{\tau^{(n)}_{j}}{\eta_n}\right)\phi_0^{(\eta_n)}\right]\ket\right.\\&\hspace{2cm}-\bbE\left.\bra J,W^{\eta_n}\left[\sum_{l=0}^ND_{\omega,l}^{\lambda}\left(\frac{\tau^{(n)}_{j}}{\eta_n}\right)\phi_0^{(\eta_n)}\right]\ket\right|\\
&\hspace{5mm}\leq\sup_{j\in\lbrace0,...,m_n\rbrace}\sum_{l_1,l_2=0}^N\left|\bra J,W^{\eta_n}\left[D_{\omega,l_1}^{\sqrt{\eta_n}}\left(\frac{\tau^{(n)}_{j}}{\eta_n}\right)\phi_0^{(\eta_n)},D_{\omega,l_2}^{\sqrt{\eta_n}}\left(\frac{\tau^{(n)}_{j}}{\eta_n}\right)\phi_0^{(\eta_n)}\right]\ket\right.\\&\hspace{4cm}-\left.\bbE\bra J,W^{\eta_n}\left[D_{\omega,l_1}^{\sqrt{\eta_n}}\left(\frac{\tau^{(n)}_{j}}{\eta_n}\right)\phi_0^{(\eta_n)},D_{\omega,l_2}^{\sqrt{\eta_n}}\left(\frac{\tau^{(n)}_{j}}{\eta_n}\right)\phi_0^{(\eta_n)}\right]\ket\right|\\
&\hspace{5mm}\leq(N+1)^2\eta_n^\gamma,
\end{split}
\end{equation}
with the last inequality holding with a probability of at least 
\begin{equation}
1-m_n(N+1)^2{C}_{J,T}^q(q-1)^{Nq}\eta_n^{(c/2-\gamma)q}.
\end{equation}
Analogously, following the proof of Lemma 3.14 in \cite{chen_rs}, we find a random variable $Y(\omega)=Y(\omega;N,n,j)$ that does not depend on $\lambda$ such that
\begin{equation}
\left\Vert R_{\omega,N}^\lambda\left(\frac{\tau^{(n)}_{j}}{\eta_n}\right)\phi_0^{\eta_n}\right\Vert^2_{\ell^2}\leq Y(\omega)
\end{equation}
for all $\lambda\in\left[0,\sqrt{\eta_n}\right]$, with $Y\in\overline{\caP_{8N}}$ and
\begin{equation}
\bbE\left[Y\right]\leq C_{T}\eta^c.
\end{equation}
By Theorem 5.10 and Remark 5.13 of \cite{janson},
\begin{equation}
\bbE\left[Y^q\right]^{1/q}\leq e^{8N}(q-1)^{4N}\bbE[Y]
\end{equation}
for all $q\geq2$.
By another Markov estimate, for any $\gamma>0$,
\begin{equation}
\label{remainderbound}
\begin{split}
&\bbP\left(\sup_{j\in\lbrace0,...,m_n\rbrace}\sup_{\lambda\in\left[0,\sqrt{\eta_n}\right]}\left\Vert R_{\omega,N}^\lambda\left(\frac{\tau^{(n)}_{j}}{\eta_n}\right)\phi_0^{\eta_n}\right\Vert_{\ell^2}\geq\eta_{n}^{\gamma}\right)\\
&\hspace{1cm}\leq m_ne^{8Nq}C_T^q(q-1)^{4Nq}\eta_{n}^{(c-2\gamma)q}.
\end{split}
\end{equation}
By (\ref{mainbound}) and (\ref{remainderbound}), if we choose $\gamma_-<\gamma<c/2$, (say, $\gamma_-=1/15$) $N$, and $q>\left(\frac{2}{\beta}+1\right)/\left(\frac{c}{2}-\gamma\right)$, we see that there is a finite constant $C_{\gamma_-,J,T}$ such that
\begin{equation}
\label{bc1}
\begin{split}
\bbP&\left(\sup_{j\in\lbrace0,...,m_n\rbrace}\sup_{\lambda\in\left[0,\sqrt{\eta_n}\right]}\left|\bra J,W^{\eta_n}\left[\mathrm{e}^{-iH^\lambda_\omega \tau^{(n)}_{j}/{\eta_n}}\phi_0^{(\eta_n)}\right]\ket\right.\right.\\ &\hspace{2cm}-\vphantom{\sup_{\lambda\in\left[0,\sqrt{\eta_n}\right]}}\left.\vphantom{\sup_{\lambda\in\left[0,\sqrt{\eta_n}\right]}}\left.\bbE\bra J,W^{\eta_n}\left[\mathrm{e}^{-iH^\lambda_\omega \tau^{(n)}_{j}/{\eta_n}}\phi_0^{(\eta_n)}\right]\ket\right|\geq\eta_{n}^{\gamma_-}\right)\leq C_{\gamma_-,J,T}n^{-3/2}.
\end{split}
\end{equation}
As we now can control the convergence of the Wigner function at selected times $\tau^{(n)}_{j}/{\eta_n}$, we turn to general $\tau/\eta\in\left[\tau^{(n)}_{j}/{\eta_n};\tau^{(n)}_{j}/{\eta_n}+\nu_n\right]$. To that end, we check that $V$ acts almost like a bounded operator. Let a small $\rho>0$ be given and choose an $L\geq1/\rho$ such that $h(y)\leq y^{-2}$ for all $y\in\bbR^3$, $|y|\geq L$. Take $R=L+T$ and define $I(x,k)$ as a test function $I:\bbR^3\times\bbT^3\rightarrow[0,1]$ to be smooth, constant in the $k$ variable with 
\begin{equation}
I(x,k)=\begin{cases}&0 \mbox{ if }|x|\geq R+1\\&1 \mbox{ if }|x|\leq R\end{cases}.
\end{equation}
All arguments so far apply to $I$ as well as $J$ and we can write
\begin{equation}
\label{propagation_bound}
\begin{split}
\sup_{j\in\lbrace0,...,m_n\rbrace}&\sup_{\lambda\in\left[\sqrt{\eta_{n+1}},\sqrt{\eta_n}\right]}\sum_{|x|\geq{(R+1)/\eta_n}}\left|\mathrm{e}^{-iH^\lambda_\omega \tau^{(n)}_{j}/{\eta_n}}\phi_0^{(\eta_n)}(x)\right|^2\\
&\leq\sup_{j\in\lbrace0,...,m_n\rbrace}\sup_{\lambda\in\left[\sqrt{\eta_{n+1}},\sqrt{\eta_n}\right]}\left(\left\Vert\phi_0^{(\eta_n)}\right\Vert_{\ell^2}^2-\bra I,W^{\eta_n}\left[\mathrm{e}^{-iH^\lambda_\omega \tau^{(n)}_{j}/{\eta_n}}\phi_0^{(\eta_n)}\right]\ket\right)\\
&\leq\sup_{\tau\in[0,T]}\bra(1-I),\mu_\tau\ket\\
&\hspace{1cm}+\sup_{\tau\in[0,T]}\sup_{\lambda\in\left[\sqrt{\eta_{n+1}},\sqrt{\eta_n}\right]}\left|\bra I,\mu_\tau\ket-\bbE\bra I,W^{\eta_n}\left[\mathrm{e}^{-iH^\lambda_\omega \tau/{\eta_n}}\phi_0^{(\eta_n)}\right]\ket\right|\\
&\hspace{1cm}+\sup_{j\in\lbrace0,...,m_n\rbrace}\sup_{\lambda\in\left[0,\sqrt{\eta_n}\right]}\left|\bra I,W^{\eta_n}\left[\mathrm{e}^{-iH^\lambda_\omega \tau^{(n)}_{j}/{\eta_n}}\phi_0^{(\eta_n)}\right]\ket\right.\\ &\hspace{5.5cm}-\left.\bbE\bra I,W^{\eta_n}\left[\mathrm{e}^{-iH^\lambda_\omega \tau^{(n)}_{j}/{\eta_n}}\phi_0^{(\eta_n)}\right]\ket\right|\\
&\leq C\rho+g_{\rho}(n)+Z_\rho(n;\omega)
\end{split}
\end{equation}
with $C$ a finite constant, a deterministic function $g_\rho$ with $g_\rho(n)\rightarrow0$ for $n\rightarrow\infty$ for any given $\rho>0$, and a random variable $Z_\rho(n)$ with 
\begin{equation}
\label{bc2}
\bbP\left(Z_\rho(n)\geq\eta_n^{\gamma_-}\right)\leq C_{\gamma_-,\rho,T}n^{-3/2}
\end{equation}
for any $\rho>0$ and $0<\gamma_-<c/2$. A similar, deterministic estimate involving only the first two summands in the last line of (\ref{propagation_bound}) is readily available for the unperturbed propagation by $H_0$.  Furthermore, there is a random variable $B$, with $\bbE\left[e^{B}\right]<\infty$, such that
\begin{equation}
\max_{|x|\leq(R+1)/\eta_n}\left|V_\omega(x)\right|\leq C_\rho B(\omega)\left|\log{\eta_n}\right|,
\end{equation}
with a $\rho$-dependent constant $C_\rho$ because of the $\rho$-dependence of $R$. Therefore, using a cutoff function $\varphi(x)=1\left\lbrace|x|\leq(R+1)/\eta_n\right\rbrace$,
\begin{equation}
\begin{split}
&\sup_{j\in\lbrace0,...,m_n-1\rbrace}\sup_{0\leq \lambda\leq\sqrt{\eta_n}}\sup_{0\leq t\leq\nu_n}\left\Vert\left(\mathrm{e}^{-iH^{\lambda}_\omega t}-\mathrm{e}^{-iH_0t}\right)\mathrm{e}^{-iH^\lambda_\omega \tau^{(n)}_{j}/{\eta_n}}\phi_0^{(\eta_n)}\right\Vert\\
&\hspace{1cm}\leq\sqrt{C\rho+g_{\rho}(n)+Z_\rho(n;\omega)}\\
&\hspace{1.5cm}+\sup_{j\in\lbrace0,...,m_n-1\rbrace}\sup_{0\leq \lambda\leq\sqrt{\eta_n}}\sup_{0\leq t\leq\nu_n}\left\Vert\lambda\int_0^t\mathrm{e}^{-iH^{\lambda}_\omega(t-s)}V\varphi\mathrm{e}^{-iH_0s}\mathrm{e}^{-iH^\lambda_\omega \tau^{(n)}_{j}/{\eta_n}}\phi_0^{(\eta_n)}\dd s\right\Vert_{\ell^2}\\
&\hspace{1cm}\leq\sqrt{C\rho+g_{\rho}(n)+Z_\rho(n;\omega)}+\sqrt{\eta_n}\nu_n|\log\eta_n|C_\rho B(\omega).
\end{split}
\end{equation}
Therefore we can interpolate between the times $\tau_j^{(n)}/\eta_n$ by replacing the full dynamics with the free one
\begin{equation}
\begin{split}
&\sup_{\eta\in\left[\eta_{n+1},\eta_{n}\right]}\sup_{\tau\in[0,T]}\left|\bra J,W^{\eta_n}\left[\mathrm{e}^{-iH_\omega\tau/\eta}\phi_0^{(\eta_n)}\right]\ket-\bbE\bra J,W^{\eta_n}\left[\mathrm{e}^{-iH_\omega\tau/\eta}\phi_0^{(\eta_n)}\right]\ket\right|\\
&\hspace{5mm}\leq\sup_{j\in\lbrace{0,...,m_n-1}\rbrace}\sup_{t\in[0,\nu_n]}\sup_{\lambda\in\left[\sqrt{\eta_{n+1}},\sqrt{\eta_n}\right]}\left|\bra J,W^{\eta_n}\left[\mathrm{e}^{-iH_0 t}\mathrm{e}^{-iH^\lambda_\omega \tau^{(n)}_{j}/{\eta_n}}\phi_0^{(\eta_n)}\right]\ket\right.\\
&\hspace{6cm}-\left.\bbE\bra J,W^{\eta_n}\left[\mathrm{e}^{-iH_0 t}\mathrm{e}^{-iH^\lambda_\omega \tau^{(n)}_{j}/{\eta_n}}\phi_0^{(\eta_n)}\right]\ket\right|\\
&\hspace{2cm}+C_J\left(\sqrt{C\rho+g_{\rho}(n)+Z_\rho(n;\omega)}+\sqrt{\eta_n}\nu_n|\log\eta_n|C_\rho B(\omega)\right).
\end{split}
\end{equation}
Directly from (\ref{defjw}), we observe
\begin{equation}
\sup_{t\in[0,\nu_n]}\left|\bra J,W^{\eta_n}\left[\mathrm{e}^{-iH_0 t}\psi\right]\ket-\bra J,W^{\eta_n}\left[\psi\right]\ket\right|\leq C_J\eta_n\nu_n\Vert\psi\Vert_{\ell^2}^2
\end{equation}
for all $\psi\in\ell^2\left(\bbZ^3\right)$, $J\in\caS\left(\bbR^3\times\bbT^3\right)$.
Up to a redefinition of constants, all estimates made so far imply
\begin{equation}
\begin{split}
&\sup_{\eta\in\left[\eta_{n+1},\eta_n\right]}\sup_{\tau\in[0,T]}\left|\bra J,W^\eta\left[\mathrm{e}^{-iH_\omega\tau/\eta}\phi_0^{(\eta)}\right]\ket-\bra J,\mu_\tau\ket\right|\\
&\hspace{1cm}\leq\sup_{\eta\in\left[\eta_{n+1},\eta_n\right]}\sup_{\tau\in[0,T]}\left|\bbE\bra J,W^\eta\left[\mathrm{e}^{-iH_\omega\tau/\eta}\phi_0^{(\eta)}\right]\ket-\bra J,\mu_\tau\ket\right|\\
&\hspace{2cm}+\sup_{j\in\lbrace{0,...,m_n}\rbrace}\sup_{\lambda\in\left[\sqrt{\eta_{n+1}},\sqrt{\eta_n}\right]}\left|\bra J,W^{\eta_n}\left[\mathrm{e}^{-iH^\lambda_\omega \tau^{(n)}_{j}/{\eta_n}}\phi_0^{(\eta_n)}\right]\ket\right.\\
&\hspace{6cm}-\left.\bbE\bra J,W^{\eta_n}\left[\mathrm{e}^{-iH^\lambda_\omega \tau^{(n)}_{j}/{\eta_n}}\phi_0^{(\eta_n)}\right]\ket\right|\\
&\hspace{2cm}+C_{h,S}C_J\left(\left(\frac{1}{\eta_{n+1}}-\frac{1}{\eta_{n}}\right)+\log\eta_n-\log\eta_{n+1}\right)\\
&\hspace{2cm}+C_J\left(\sqrt{C\rho+g_{\rho}(n)+Z_\rho(n;\omega)}+\sqrt{\eta_n}\nu_n|\log\eta_n|C_\rho B(\omega)\right)\\
&\hspace{2cm}+C_J\eta_n\nu_n.
\end{split}
\end{equation}
If we choose $\eta_n=n^{-\beta}$ with $\beta\in(0,1)$ and note that $\nu_n=\caO(1)$ by this choice, a standard Borel-Cantelli argument based on (\ref{bc1}) and (\ref{bc2}) yields $\bbP$-almost surely
\begin{equation}
\limsup_{\eta\rightarrow0}\sup_{\tau\in[0,T]}\left|\bra J,W^\eta\left[\mathrm{e}^{-iH_\omega\tau/\eta}\phi_0^{(\eta)}\right]\ket-\bra J,\mu_\tau\ket\right|\leq C_J\sqrt{C\rho}
\end{equation}
for arbitrary $\rho>0$, which together with the separability of $\caS\left(\bbR^3\times\bbT^3\right)$ proves the claim.

\appendix
\section{Basic estimates}
The following is a small extension and correction of Lemma 3.4 in \cite{chen_rs}
\begin{lma}
With the definitions made as above,  the estimates
\begin{equation}
\label{intest2d}
\sup_{\alpha\in I}\sup_{p_3\in[0,1)}\int_{\ttwo}\dd\und{p}\frac{1}{\left|e_{2D}(\und{p})-\alpha(p_3)-i\eps\right|}\leq C|\log\eps|^2
\end{equation}
and
\begin{equation}
\label{intest3d}
\sup_{\alpha\in I}\int_{\tthree}\dd{p}\frac{1}{\left|e({p})-\alpha-i\eps\right|}\leq C|\log\eps|
\end{equation}
hold for all $0<\eps\leq1/3$.
\end{lma}
\begin{proof}
We first verify (\ref{intest2d}). By definition, for any $p_3$, $\alpha(p_3)\in[\alpha-4,\alpha-2]$ and $\alpha\in I$. Now for small $\eps>0$, $\left|e_{2D}(\und{p})-\alpha(p_3)-i\eps\right|>1-\eps$ whenever $\alpha\in I$ is not real, so that case is trivial. We thus concentrate on real $\alpha$, in which case $\rho=\alpha(p_3)\in[-5,5]$. For fixed $r\in\bbR$, we define the level set $s_r$ as 
\begin{equation}
s_r=\left\lbrace\und{p}\in\ttwo:e_{2D}(\und{p})=r\right\rbrace
\end{equation}
and then consider the quantity
\begin{equation}
\int_{\ttwo}\dd \und{p}\delta\left({e_{2D}(\und{p})-r}\right).
\end{equation}
This integral is not, as claimed in \cite{chen_rs}, bounded in $r$. Note that it is not equal to the arc length of $s_r$ (which is bounded), but rather the integral of $\left|\nabla e_{2D}\right|^{-1}$ along $s_r$, which diverges for the square-shaped $s_0$. Still, by controlling how close $s_r$ gets to the zeroes of $\left|\nabla e_{2D}\right|$ for given $r$, the estimate
\begin{equation}
\label{deltaint2d}
\int_{\ttwo}\dd \und{p}\delta\left({e_{2D}(\und{p})-r}\right)\leq\begin{cases}C\left(1+|\log|r||\right) &\mbox{if } r\in[-2,2]\\0 &\mbox{if } r\notin[-2,2]\end{cases}
\end{equation}
is immediate. Now, similar to \cite{chen_rs} let for real $\rho\in[-5,5]$ and all $n\in\bbN_0$ with $2^n\eps\leq 5$
\begin{equation}
\label{defan}
A_n(\eps,\rho)=\left\lbrace\und{p}\in\ttwo:2^n\eps\leq\left|e_{2D}(\und{p})-\rho\right|\leq2^{n+1}\eps\right\rbrace
\end{equation}
with two-dimensional Lebesgue measure
\begin{equation}
\begin{split}
\left|A_n(\eps,\rho)\right|&\leq\left(\int_{-2^{n+1}\eps}^{-2^n\eps}\dd \rho^{'}+\int_{2^{n}\eps}^{2^{n+1}\eps}\dd \rho^{'}\right)\int_{\ttwo}\dd \und{p}\delta\left({e_{2D}(\und{p})-\rho-\rho^{'}}\right)\\
&\leq C|\log\eps|2^n\eps,
\end{split}
\end{equation}
and by the same reasoning we have the estimate for the two-dimensional Lebesgue measure
\begin{equation}
\left|\left\lbrace\und{p}\in\ttwo:\left|e_{2D}(\und{p})-\rho\right|\leq\eps\right\rbrace\right|\leq C\eps|\log\eps|,
\end{equation}
with the constant $C$ here, as always, not depending on any of the parameters. Therefore (constantly redefining $C$), we have
\begin{equation}
\label{proofintest2d}
\begin{split}
\int_{\ttwo}\dd\und{p}\frac{1}{\left|e_{2D}(\und{p})-\rho-i\eps\right|}&\leq C\frac{\eps|\log\eps|}{\eps}+\sum_{n=0}^{C|\log\eps|}\frac{\left|A_n(\eps,\rho)\right|}{2^n\eps}\\
&\leq C|\log\eps|^2,
\end{split}
\end{equation}
as desired. For the three-dimensional case, we again only show the case of real $\alpha$, for which by (\ref{deltaint2d})
\begin{equation}
\label{deltaint3d}
\begin{split}
\int_{\tthree}\dd{p}\delta\left({e({p})-\alpha}\right)&=\int_0^1\dd p_3\int_{\ttwo}\dd\und{p}\delta\left(e_{2D}(\und{p})+3-\cos{2\pi p_3}-\alpha\right)\\
&\leq \tilde{C}\int_{-1}^1\dd k(1-k^2)^{-1/2}\left(1+\left|\log\left|3-k-\alpha\right|\right|\right)\\
&\leq C
\end{split}
\end{equation}
with $C$ independent of $\alpha$. Now repeating steps (\ref{defan}) through (\ref{proofintest2d}) yields (\ref{intest3d}).
\end{proof}
\begin{defin}
\label{bilinear}
To compactify notation, we define a bilinear extension of the $\eta$-scaled Wigner transform
\begin{equation}
W^{\eta}\left[\phi,\psi\right](X,V)=\sum_{\substack{y,z\in\bbZ^3\\y+z=2X/\eta}}\overline{\phi(y)}\psi(x)e^{2\pi i V\cdot(y-z)}
\end{equation}
for all $\phi,\psi\in\ell^2\left(\bbZ^3\right)$.
\end{defin}
\begin{lma}
\label{est_bilin}
For all $\eta>0$,
\begin{equation}
\left|\bra J,W^{\eta}\left[\phi,\psi\right]\ket\right|\leq\int_{\bbR^3}\dd\xi\sup_{v\in\tthree}\left|\hat{J}(\xi,v)\right|\Vert\phi\Vert_{\ell^2}\Vert\psi\Vert_{\ell^2}=C_J\Vert\phi\Vert_{\ell^2}\Vert\psi\Vert_{\ell^2}.
\end{equation}
\end{lma}
\begin{proof}
Directly from the Fourier representation
\begin{equation}
\bra J,W^{\eta}\left[\phi,\psi\right]\ket=\int_{\bbR^3\times\tthree}\dd\xi\dd v\widehat{J_\eta}(\xi,v)\overline{\widehat{\phi}(v-\xi/2)}\widehat{\psi}(v+\xi/2).
\end{equation}
\end{proof}
\label{app_basic}
\section{Two-resolvent integral}
\label{app_2res}
\begin{lma}
\label{lemma_2res}
There is a finite set $\caE=\left((0,0,0),\left(\frac{1}{2},\frac{1}{2},\frac{1}{2}\right)\right)\subset\tthree$ and a constant $C<\infty$ such that
\begin{equation}
\begin{split}
\sup_{\alpha,\beta\in I}&\int_{\tthree}\dd u \frac{1}{\left|e(u+p)-\alpha-i\eps\right|\left|e(u)-\beta-i\eps\right|}\\&\leq C\left( \eps^{-4/5}+\frac{\eps^{-2/5}}{\dist(p,\caE)}\right)\left|\log\eps\right|^2
\end{split}
\end{equation}
for all $p\in\bbT^3$ and $0<\eps\leq1/3$.
\end{lma}

\begin{proof}
To have a more symmetric integrand, we choose $k\in\tthree$ such that $2k=p$\\
$(\mod\tthree)$ and shift the integrand by $k$, so we want to estimate
\begin{equation}
\begin{split}
\sup_{\alpha,\beta\in I}&\int_{\tthree}\dd u \frac{1}{\left|e(u+k)-\alpha-i\eps\right|\left|e(u-k)-\beta-i\eps\right|}
\end{split}
\end{equation}
Decomposing three-dimensional vectors like $p=\left(\und{p},p_3\right)=\left(p_1,p_2,p_3\right)$ we have
\begin{equation}
\begin{split}
&\int_{\tthree}\dd u \frac{1}{\left|e(u+k)-\alpha-i\eps\right|\left|e(u-k)-\beta-i\eps\right|}\\
&=\int_0^1\dd u_3\int_{\ttwo}\dd\und{u}\frac{1}{\left|e_{2D}(\und{u}+\und{k})-\alpha(u_3+k_3)-i\eps\right|\left|e_{2D}(\und{u}-\und{k})-\beta(u_3-k_3)-i\eps\right|}
\end{split}
\end{equation}
with the shorthand
\begin{equation}
e_{2D}(\und{p})=-\sum_{j=1}^2\cos{2\pi p_j}
\end{equation}
and
\begin{equation}
\begin{split}
\alpha(p_3)&=\alpha-3+\cos{2\pi p_3}\\
\beta(p_3)&=\beta-3+\cos{2\pi p_3}.
\end{split}
\end{equation}
Next, we want to transform the integral over $\ttwo$ by writing $\und{u}$ as the intersection of two level sets $e_{2D}(\und{u}+\und{k})=\rho$ and $e_{2D}(\und{u}-\und{k})=\tilde{\rho}$. By elementary algebra, for given $\rho,\tilde{\rho}\in[-2,2]$ we will at most find four such intersections $\und{u}$. The only exception are the degenerate cases in which any of the following holds
\begin{align}
\label{degenerate_intersection1}
2\und{k}&=0\mod{\ttwo}& &\mbox{and} & \rho&=\tilde{\rho},\\
2\und{k}&=\left(\frac{1}{2},\frac{1}{2}\right)\mod{\ttwo}& &\mbox{and} & \rho&=-\tilde{\rho},\\
2k_1&=2k_2\mod[0,1)& &\mbox{and} & \rho&=\tilde{\rho}=0\hspace{5mm}\mbox{ or }\\
\label{degenerate_intersection4}
2k_1&=-2k_2\mod[0,1)& &\mbox{and} & \rho&=\tilde{\rho}=0.
\end{align}
The transformation has the Jacobian determinant
\begin{equation}
\begin{split}
J(\und{u})&=4\pi^2\det\begin{pmatrix} \sin(2\pi(u_1+k_1)) & \sin(2\pi(u_1-k_1))\\ \sin(2\pi(u_2+k_2)) & \sin(2\pi(u_2-k_2))\end{pmatrix}\\
&=8\pi^2\sin(2\pi k_1)\cos(2\pi k_2)\cos(2\pi u_1)\sin(2\pi u_2)\\&\hspace{1cm}-8\pi^2\sin(2\pi k_2)\cos(2\pi k_1)\cos(2\pi u_2)\sin(2\pi u_1)\\
&=4\pi^2\begin{pmatrix}\sin(2\pi(u_1+u_2))\\\sin(2\pi(u_1-u_2))\end{pmatrix}\cdot\begin{pmatrix}\sin(2\pi(k_1-k_2))\\-\sin(2\pi(k_1+k_2))\end{pmatrix}.
\end{split}
\end{equation}
We now have to control the cases where $|J|$ is small. Note that this will automatically exclude the degeneracy from (\ref{degenerate_intersection1}-\ref{degenerate_intersection4}) as $J(\und{u})=0$ there. With 
\begin{equation}
\tilde{\caE}=\left\lbrace\und{k}\in\ttwo: k_1\in\left\lbrace0,\frac{1}{4},\frac{1}{2},\frac{3}{4}\right\rbrace; k_2\in\left\lbrace k_1,k_1+\frac{1}{2}\right\rbrace\right\rbrace
\end{equation}
the set of zeroes of the second vector in the scalar product, there is a normalized vector $\und{n}=(n_1,n_2)$ depending on $\und{k}$ such that
\begin{equation}
\begin{pmatrix}\sin(2\pi(k_1-k_2))\\-\sin(2\pi(k_1+k_2))\end{pmatrix}=\theta\und{n}
\end{equation}
with
\begin{equation}
\theta\geq c\,\dist\left(\und{k},\tilde{\caE}\right)
\end{equation}
for some positive constant $c$. For a $0<\delta\ll1$ to be optimized later the set
\begin{equation}
A_\delta=\left\lbrace\und{u}\in\ttwo:\left|\begin{pmatrix}\sin(2\pi(u_1+u_2))\\\sin(2\pi(u_1-u_2))\end{pmatrix}\cdot\und{n}\right|<\delta\right\rbrace
\end{equation}
has a two-dimensional Lebesgue measure $\left|A_\delta\right|<C\delta|\log\delta|$ for some finite constant $C$. Keeping in mind that the location of $A_\delta$ depends on $k_1$ and $k_2$, but not on $u_3$ and $k_3$, we first estimate the integral
\begin{equation}
\label{cauchyfactor}
\int_{\tthree}\dd u \frac{1_{A_\delta}(\und{u})}{\left|e(u+k)-\alpha-i\eps\right|^2}
\end{equation}
and a similar one for $-k$ and $\beta$ instead of $k$ and $\alpha$. For $0<{h}\ll1$, one only has $\left|\frac{\partial}{\partial u_3}\alpha(u_3+k_3)\right|<{h}$ for $u_3\in K_{h}$ with one-dimensional Lebesgue measure $\left| K_{h} \right|<C{h}$, and thus
\begin{equation}
\begin{split}
(\ref{cauchyfactor})&\leq\int_{K_{h}}\dd u_3\int_{\ttwo}\dd\und{u}\frac{1}{\left|e_{2D}(\und{u}+\und{k})-\alpha(u_3+k_3)-i\eps\right|^2}\\&\hspace{2cm}+\int_{A_\delta}\dd\und{u}\int_{K_{h}^c}\dd u_3\frac{1}{\left|e_{2D}(\und{u}+\und{k})-\alpha(u_3+k_3)-i\eps\right|^2}\\
&\leq C\left(\frac{{h}}{\eps}|\log{\eps}|^2+\frac{\delta|\log\delta|}{\eps{h}}\right),
\end{split}
\end{equation}
where we used the $L^1$ estimate (\ref{intest2d}) and the $L^\infty$ estimate $1/\eps$ for the resolvents in the first summand. By Cauchy-Schwarz, 
\begin{equation}
\label{adeltaint}
\int_{\tthree}\dd u \frac{1_{A_\delta}(\und{u})}{\left|e(u+k)-\alpha-i\eps\right|\left|e(u-k)-\beta-i\eps\right|}\leq C\left(\frac{{h}}{\eps}|\log{\eps}|^2+\frac{\delta|\log\delta|}{\eps{h}}\right)
\end{equation}
as well.

Now we are ready for the integral transform. Let a $\delta>0$ and $\theta>0$ be given (our estimate will yield infinity for $\und{k}\in\tilde{\caE}$). Then for $\und{u}\in A_\delta^c=\ttwo\setminus A_\delta$, $|J(\und{u})|>c\theta\delta$ for some constant $c>0$, and we can locally transform the $\dd\und{u}$ integral to a $\dd\rho\dd\tilde{\rho}$ integral. As $A_\delta^c$ is compact, we can cover it by finitely many local coordinate patches, and as remarked above, the same coordinates $(\rho,\tilde{\rho})$ will be asigned to at most four different points $\und{u}$ in this procedure. Thus
\begin{equation}
\begin{split}
\int_{A_\delta^c}&\dd\und{u}\frac{1}{\left|e_{2D}(\und{u}+\und{k})-\alpha(u_3+k_3)-i\eps\right|\left|e_{2D}(\und{u}-\und{k})-\beta(u_3-k_3)-i\eps\right|}\\
&\leq\frac{\tilde{C}}{\delta\theta}\int_{-2}^2\dd\rho\int_{-2}^2\dd\tilde{\rho}\frac{1}{\left|\rho-\alpha(u_3+k_3)-i\eps\right|\left|\tilde{\rho}-\beta(u_3-k_3)-i\eps\right|}\\
&\leq\frac{C}{\delta\theta}\left|\log\eps\right|^2.
\end{split}
\end{equation}

Integrating over $u_3$ and collecting all estimates, we arrive at
\begin{equation}
\begin{split}
\sup_{\alpha,\beta\in I}&\int_{\tthree}\dd u \frac{1}{\left|e(u+k)-\alpha-i\eps\right|\left|e(u-k)-\beta-i\eps\right|}\\&\leq C\left(\frac{{h}}{\eps}|\log{\eps}|^2+\frac{\delta|\log\delta|}{\eps{h}}+\frac{1}{\delta\dist(\und{k},\tilde{\caE})}\left|\log\eps\right|^2\right).
\end{split}
\end{equation}

Singling out the $k_3$ component was an arbitrary choice, so we even obtain
\begin{equation}
\label{est_hdeltaeps}
\begin{split}
\sup_{\alpha,\beta\in I}&\int_{\tthree}\dd u \frac{1}{\left|e(u+k)-\alpha-i\eps\right|\left|e(u-k)-\beta-i\eps\right|}\\&\leq C\left(\frac{{h}}{\eps}|\log{\eps}|^2+\frac{\delta|\log\delta|}{\eps{h}}+\frac{1}{\delta\dist({k},\widehat{\caE})}\left|\log\eps\right|^2\right).
\end{split}
\end{equation}
with 
\begin{equation}
\widehat{\caE}=\left\lbrace k\in\tthree : \left(k_i,k_j\right)\in\tilde{\caE}\forall i,j\right\rbrace,
\end{equation}
and note that 
\begin{equation}
\widehat{\caE}=\left\lbrace k\in\tthree : 2k\in\caE\right\rbrace.
\end{equation}
Redefining $C$ and optimizing for later application, we set $\delta={h}^2=\eps^{2/5}$ and arrive at
\begin{equation}
\begin{split}
\sup_{\alpha,\beta\in I}&\int_{\tthree}\dd u \frac{1}{\left|e(u+k)-\alpha-i\eps\right|\left|e(u-k)-\beta-i\eps\right|}\\&\leq C\left( \eps^{-4/5}+\frac{\eps^{-2/5}}{\dist({2k},{\caE})}\right)\left|\log\eps\right|^2.
\end{split}
\end{equation}
\end{proof}
This two-resolvent estimate also allows for a simpler proof of a slightly sharper estimate for the three-resolvent-integral (\ref{three_res_int}). We still continue to use the original estimate from \cite{chen_rs} in the proof of the main theorem as any estimate of the kind $\eps^{-1+\delta}$ suffices.
\begin{cor}
There is a constant $C<\infty$ such that for all $0<\eps\leq\frac{1}{3}$,
\begin{equation}
\label{three_res_improved}
\sup_{\gamma_i\in I}\sup_{k\in\tthree}\int_{\left(\tthree\right)^2}\frac{\dd p\dd q}{\left|e(p)-\gamma_1-i\eps\right|\left|e(q)-\gamma_2-i\eps\right|\left|e(p\pm q+k)-\gamma_3-i\eps\right|}\leq C\eps^{-7/9}\left|\log\eps\right|^4.
\end{equation}
\end{cor}
\begin{proof}
Without loss of generality, choose the sign $p+q+k$ on the left side of (\ref{three_res_improved}) and first take the $q$ integral. We have the estimates
\begin{equation}
\int_{\tthree}\frac{\dd q}{\left|e(q)-\gamma_2-i\eps\right|\left|e(p+ q+k)-\gamma_3-i\eps\right|}\leq\begin{cases}C&\left(\eps^{-7/9}+\frac{\eps^{-4/9}}{\dist(p+k,\caE)}\right)|\log\eps|^2 \\C&\frac{|\log\eps|}{\eps}, \end{cases}
\end{equation}
with the first line resulting from plugging $\delta=h^2=\eps^{4/9}$ into (\ref{est_hdeltaeps}), and the second line being the trivial estimate. To apply a H\"{o}lder inequality, note that
\begin{equation}
\left\Vert\min\left\lbrace\frac{|\log\eps|}{\eps},\frac{\eps^{-4/9}|\log\eps|^2}{\dist(\cdot+k,\caE)}\right\rbrace\right\Vert_{L^3\left(\tthree\right)}\leq C\eps^{-4/9}|\log\eps|^3,
\end{equation}
while
\begin{equation}
\left(\int_{\tthree}\frac{\dd p}{\left|e(p)-\gamma_1-i\eps\right|^{3/2}}\right)^{\frac{2}{3}}\leq C\eps^{-1/3}|\log\eps|.
\end{equation}
Collecting all estimates, we obtain
\begin{equation}
\begin{split}
\sup_{\gamma_i\in I}\sup_{k\in\tthree}&\int_{\left(\tthree\right)^2}\frac{\dd p\dd q}{\left|e(p)-\gamma_1-i\eps\right|\left|e(q)-\gamma_2-i\eps\right|\left|e(p\pm q+k)-\gamma_3-i\eps\right|}\\&\leq C\left(\eps^{-7/9}\left|\log\eps\right|^3+\eps^{-7/9}\left|\log\eps\right|^4\right),
\end{split}
\end{equation}
and redefining $C$ proves the corollary.
\end{proof}

\bibliography{references_variance}{}
\bibliographystyle{plain}

\end{document}